\documentclass[11pt]{article}

\usepackage[margin=1in]{geometry}
\usepackage{amsmath,amssymb,amsthm,mathtools}
\usepackage{aliascnt}
\usepackage{bbm}
\usepackage{enumitem}
\usepackage{array,booktabs}
\usepackage{placeins}
\usepackage{xcolor}
\usepackage[colorlinks=true,linkcolor=magenta,citecolor=magenta,urlcolor=magenta]{hyperref}
\usepackage{microtype}
\usepackage[basic]{complexity}
\usepackage{tikz}
\usetikzlibrary{arrows.meta,positioning,fit,calc,shapes.misc,backgrounds}

\usepackage[numbers]{natbib}

\usepackage[nameinlink,capitalise,noabbrev]{cleveref}

\theoremstyle{plain}
\newtheorem{theorem}{Theorem}[section]
\newaliascnt{lemma}{theorem}
\newtheorem{lemma}[lemma]{Lemma}
\aliascntresetthe{lemma}
\newaliascnt{proposition}{theorem}
\newtheorem{proposition}[proposition]{Proposition}
\aliascntresetthe{proposition}
\newaliascnt{corollary}{theorem}
\newtheorem{corollary}[corollary]{Corollary}
\aliascntresetthe{corollary}
\theoremstyle{definition}
\newaliascnt{definition}{theorem}
\newtheorem{definition}[definition]{Definition}
\aliascntresetthe{definition}
\theoremstyle{remark}
\newaliascnt{remark}{theorem}

\aliascntresetthe{remark}
\theoremstyle{plain}
\newaliascnt{conjecture}{theorem}

\aliascntresetthe{conjecture}

\crefname{theorem}{Theorem}{Theorems}
\Crefname{theorem}{Theorem}{Theorems}
\crefname{lemma}{Lemma}{Lemmas}
\Crefname{lemma}{Lemma}{Lemmas}
\crefname{proposition}{Proposition}{Propositions}
\Crefname{proposition}{Proposition}{Propositions}
\crefname{corollary}{Corollary}{Corollaries}
\Crefname{corollary}{Corollary}{Corollaries}
\crefname{definition}{Definition}{Definitions}
\Crefname{definition}{Definition}{Definitions}
\crefname{remark}{Remark}{Remarks}
\Crefname{remark}{Remark}{Remarks}
\crefname{conjecture}{Conjecture}{Conjectures}
\Crefname{conjecture}{Conjecture}{Conjectures}
\crefname{equation}{Equation}{Equations}
\Crefname{equation}{Equation}{Equations}

\newcommand{\GammaPI}{\Gamma^{\operatorname{PI}}}
\newcommand{\GammaPR}{\Gamma^{\operatorname{PR}}}
\newcommand{\defeq}{\coloneqq}
\newcommand{\OPT}{\operatorname{OPT}}
\newcommand{\Ind}{\operatorname{Ind}}
\newcommand{\VR}{\operatorname{VR}}
\newcommand{\pistrat}{{\boldsymbol{\pi}}}
\newcommand{\hatpistrat}{{\widehat{\boldsymbol{\pi}}}}
\newcommand{\bstrat}{{\boldsymbol{b}}}

\usepackage{thm-restate}

\title{What You Earn from What You Learn:\\
Characterization and Complexity of the Value of Information}
\author{%
Ratip Emin Berker\textsuperscript{1,2}\thanks{These authors contributed equally to this work.}\quad
Ioannis Anagnostides\textsuperscript{1}\footnotemark[1]\quad
Emanuel Tewolde\textsuperscript{1,2}\\
Brian Hu Zhang\textsuperscript{3}\quad
Tuomas Sandholm\textsuperscript{1}\quad
Vincent Conitzer\textsuperscript{1,2}\\[0.5em]
\small \textsuperscript{1}Carnegie Mellon University, Pittsburgh PA 15213, USA\\
\small \textsuperscript{2}Foundations of Cooperative AI Lab (FOCAL)\\
\small \textsuperscript{3}Massachusetts Institute of Technology, Cambridge MA 02139, USA\\
\small \texttt{\{rberker,ianagnos,etewolde,sandholm,conitzer\}@cs.cmu.edu}\\
\small \texttt{zhangbh@csail.mit.edu}}
\date{\today}

\begin{document}
\pagenumbering{roman}
\maketitle

\begin{abstract}
We examine the \emph{value of information} in extensive-form decision problems: the largest multiplicative gain that a decision maker can obtain from an \emph{information refinement}, taken over all payoff environments. This measures the structural importance of information disentangled from any particular payoff specification, and generalizes the notion of {\em value of recall} recently introduced by Berker et al. (AAAI ‘25). We establish a unifying characterization of the value of information under any information refinement, uncovering an intriguing connection to the maximum \emph{fractional chromatic number} over a certain family of conflict graphs. We then employ our characterization to elucidate the complexity of two central information refinements.

First, for the \emph{value of perfect information} in perfect-recall decision problems, we find that the relevant conflict graphs are cographs, so the fractional chromatic number collapses to the clique number. Even though the clique number is tractable for a given cograph, we show that the additional layer of optimization (over the family of graphs) makes computing the value of perfect information $\mathsf{NP}$-complete. On the positive side, we prove fixed-parameter tractability in terms of the \emph{degree} of imperfect information. Second, for the value of recall in imperfect-recall decision problems (with no absentmindedness), we first show that the value of recall can realize the fractional chromatic number of any graph. Interestingly, this implies that, unlike the value of perfect information, the value of recall can take fractional values. Building on this reduction, we then establish $\mathsf{\Sigma^p_2}$-completeness. This addresses questions left open by Berker et al. (AAAI ‘25)---both on the characterization and the complexity of the value of recall---and shows that computing the value of recall is strictly harder than the value of perfect information.

Taken together, our results shed new light on decision making under information constraints, which is at the heart of economics and optimization.
\end{abstract}

\clearpage
\tableofcontents

\clearpage
\pagenumbering{arabic}

\section{Introduction}

Few problems have commanded as much sustained attention in economics as the problem of information. Its role has permeated microeconomics---so much so that the entire branch of~\emph{information economics} revolves around it---and underpins several pathbreaking developments: Akerlof's incisive analysis of markets with asymmetric information~\citep{Akerlof70:Market}, and the resulting theory of \emph{adverse selection}, where one side of a transaction possesses information the other lacks; \emph{moral hazard} and the crystallization of contract theory~\citep{Grossman83:Analysis}, where an agent's hidden actions after a contract is signed may impose adverse consequences on the principal; Spence's theory of \emph{signaling}~\citep{Spence73:Job} by which an informed party can credibly convey information; Stiglitz's pioneering analysis of~\emph{screening}~\citep{Stiglitz75:Theory}, whereby the underinformed party can induce the other party to reveal private information; and \emph{mechanism design}~\citep{Myerson08:Mechanism}, which studies how to design allocation rules to elicit and aggregate agents' private information. Indeed, the 2001 Nobel Prize in economics~\citep{Nobel01:Asymmetric} was awarded to Akerlof, Spence, and Stiglitz ``for their analyses of markets with asymmetric information.'' More recently, the control of information has become a central object of study in \emph{Bayesian persuasion}~\citep{Kamenica11:Bayesian,Bergemann19:Information}, where an informed sender commits to an information structure in order to influence a receiver's action. These questions are also becoming increasingly central in theoretical computer science, shaping the study of information design~\citep{Dughmi16:Algorithmic,Dughmi19:Persuasion,Celli20:Private,Castiglioni20:Online,Su2021:Bayesian,Yang24:Computational,Babichenko25:Information,Chen25:Explainable,Duetting25:Information,Ding23:Competitive,Guo25:Value}, algorithmic contract theory~\citep{Dutting24:Algorithmic,Gan22:Generalized,Alon23:Bayesian,Castiglioni25:Reduction}, and strategic machine learning~\citep{Hardt16:Strategic,Cai15:Optimum,Caragiannis16:Truthful,Feng22:Bias,Kong20:Information}.

\emph{Extensive-form} decision problems provide a canonical formalism for capturing the role of information, encoding not just the order of moves and the available actions, but also the information available at each decision point. The decision maker may lack information because it was never observed (such as the cards of other players in poker), or because it was observed but was later forgotten. The former concerns \emph{imperfect information}; the latter \emph{imperfect recall}.

While not as widely studied as other forms of imperfect information, the role of recall has still been a mainstay of game theory tracing back to the foundational work of~\citet{Kuhn53:Extensive} in the 1950s. Research on imperfect recall has recently been invigorated by many applications ranging from \emph{game abstraction}~\citep{Waugh09:Abstraction,Kroer14:Extensive,Kroer16:Imperfect,Lanctot12:No,Heymann24:Learning}, \emph{team games} with limited communication~\citep{VonStengel97:Team,Basilico17:Team,Zhang22:Correlation}, and problems in trustworthy AI and AI safety \citep{Tewolde26:Decision}, where imperfect recall arises in connection to simulating agents~\citep{Conitzer23:Foundations,Pan23:Do,Kovarik25:AI} and privacy concerns~\citep{Arrow62,Stephenson25:NDAI}, exemplified by the EU's ``right to be forgotten''~\citep{gdpr}. Against this backdrop, a key reference point for our work is the recent framework of~\citet{Berker25:Value}, who introduced the \emph{value of recall} of an imperfect-recall decision problem as the ratio between the utility that a player could obtain if endowed with perfect recall and the utility they can obtain in the original problem.

In this paper, building on~\citet{Berker25:Value}, we introduce a broader formulation that captures the \emph{value of information}, with the value of recall arising as a special case. Specifically, we are given as input an extensive-form decision problem, specifying a tree, chance probabilities, and an \emph{information partition}, which groups together nodes the decision maker cannot distinguish; the elements of this partition are known as~\emph{information sets}. We are additionally given an \emph{information refinement}, obtained by splitting some of the original information sets, hence enhancing the decision maker's ability to distinguish nodes. The value of information measures how much the decision maker \emph{could} gain from the added information. Following prior work on related problems in algorithmic game theory~\citep{Koutsoupias99:Worst,Anshelevich08:Price,Value08:Value,Letchford14:Value}, we define this gain as a \emph{ratio}. Unlike the ground model of \citet{Berker25:Value}, payoffs are \emph{not} part of the input; instead, we let the value of information be the largest possible gain over \emph{all} (nonnegative) payoff environments. This allows us to disentangle the structural importance of the information partition from the contribution of a particular payoff specification.

This perspective is motivated by many applications in which the information structure is the primary object of interest, whereas the payoff specification is uncertain, ephemeral, or application-dependent. For example, a privacy constraint may be fixed before one knows the exact objective with respect to which the resulting decision problem will be evaluated, or a mechanism designer may need to construct an information structure that performs well across many possible agent preferences. Similarly, in settings involving humans or large language models, the relevant payoff function may be difficult to specify or unstable across contexts. In such cases, a small value of information provides a strong guarantee: no matter what payoff environment is ultimately realized, the chosen refinement can improve the decision maker's value by only a limited amount. Conversely, if the value of information is large, then the payoff environment witnessing this large gap identifies the kinds of objectives for which the added information is most consequential. Overall, in many economic settings, the information structure is more stable than the payoff environment.

\subsection{Our results}
\label{sec:contributions}

We provide a unifying characterization together with new complexity results for the value of information in extensive-form decision problems. Our approach revolves around an intriguing connection we make between the value of information and the maximum \emph{fractional chromatic number} over a certain family of \emph{conflict graphs} (\Cref{def:conflict-graph}) induced by the information structure (\Cref{thm:fractional-formula}). We then leverage this general graph-theoretic characterization to shed light on the complexity of two important information refinements: i) the value of \emph{perfect information} in perfect-recall decision problems and ii) the value of \emph{recall} in imperfect-recall decision problems.

\paragraph{Value of perfect information.} First, for the value of perfect information, we show that every conflict graph is a \emph{cograph} (\Cref{lem:cograph}). As a result, the fractional chromatic number collapses to the clique number, yielding a more refined characterization. Interestingly, this shows that the value of perfect information is always an integer. Moreover, even though the clique number is tractable for a given cograph, we establish that the additional layer of optimization (over the family of graphs) makes the problem $\mathsf{NP}$-complete through a reduction from \textsc{MaximumCoverage}.

\begin{theorem}
    Deciding whether the value of perfect information in perfect-recall decision problems is at least a threshold $\tau \in \mathbb{Q}$ is \NP-complete.
\end{theorem}

Moreover, our reduction is exact (\Cref{thm:voi-coverage-exact}), thereby also implying $\class{APX}$-hardness (\Cref{cor:voi-approx-hardness}) by virtue of the well-known inapproximability of \textsc{MaximumCoverage}~\citep{Feige98:Threshold}.

On the positive side, we show that the value of perfect information can be computed in polynomial time via dynamic programming if the size of each information set---which can be thought of as the \emph{degree} of imperfect information---is bounded (\Cref{thm:bounded-width-dp}).

\paragraph{Value of recall.} For the value of recall, we begin by constructing a class of decision problems whose value of recall is precisely equal to the fractional chromatic number of any input graph (\Cref{thm:graph-embedding}), thereby implying \coNP-hardness and inapproximability (\Cref{cor:VoR-inapprox}). This reduction is particularly natural in view of our unifying characterization based on fractional coloring (\Cref{thm:fractional-formula}), and shows that---unlike the value of perfect information---the value of recall can take fractional (non-integer) values.

Building on this reduction, we establish that, as before, the additional layer of optimization over the conflict graphs lifts the complexity in the polynomial hierarchy, now yielding $\class{\Sigma^p_2}$-completeness.

\begin{theorem}
    Deciding whether the value of recall of imperfect-recall decision problems with no absentmindedness\footnote{Throughout this paper, we assume the decision problems are free of absentmindedness, which is a particular form of imperfect recall wherein the same information set can be entered multiple times on a path (\Cref{def:no-absentmindedness}).} is at least a threshold $\tau \in \mathbb{Q}$ is $\class{\Sigma^p_2}$-complete.
\end{theorem}

This gives a complete characterization of its complexity, addressing a key question left open by~\citet{Berker25:Value}. They observed that for decision problems with absentmindedness but no chance nodes there is a simple (polynomial-time) characterization for the value of recall~\citep[Corollary 25]{Berker25:Value}. However, they did not obtain a tight characterization for the more challenging case where chance nodes are present, which is the crux of this problem, deriving only a loose bound~\citep[Corollary 27]{Berker25:Value}. Our paper closes this gap both in terms of providing a characterization and resolving its complexity.

From a technical standpoint, our reduction establishing $\class{\Sigma^p_2}$-hardness is especially natural, and builds on the aforementioned reduction based on fractional coloring. In particular, the reduction route goes from the canonical $\class{\Sigma^p_2}$-complete problem $\exists \forall$ DNF satisfiability, to 3-colorability \emph{with optional edges} (\Cref{lem:optional-edge-3col}), fractional colorability again with optional edges (\Cref{lem:fglss-gap}), and finally to the value of recall (\Cref{lem:vor-graph-realization}). The basic idea in this last step is that we start from our previous reduction based on fractional coloring, and then construct on top of that additional \emph{selection} gadgets, which allow us to encode the presence of optional edges.

Taken as a whole, the main takeaways from our paper are as follows.

\begin{itemize}
    \item Determining the value of information is hard even in ostensibly simple environments. In particular, since our hardness results only require a single decision maker, they certainly imply intractability for games---settings with multiple players. Also, our results apply even under perfect recall (for value of perfect information) or no absentmindedness (for value of recall). At the same time, more structured decision problems elude these barriers, as shown by our fixed-parameter tractability result in terms of the degree of imperfect information.
    \item Computing the value of perfect information is strictly easier than the value of recall (unless there is a complexity collapse in the polynomial hierarchy), demonstrating that there are fundamental separations between different types of information refinements. From a more practical standpoint, problems in \NP~are amenable to highly effective tools such as SAT solvers, whereas our results show that such techniques are not applicable to value of recall.
    \item Our characterization in terms of fractional chromatic number applies to \emph{any} information refinement and uncovers a new connection between imperfect-information decision making and graph-theoretic notions.
\end{itemize}

\paragraph{Illustrative example.} In Akerlof's ``market for lemons''~\citep{Akerlof70:Market}, a seller observes a car's quality while a prospective buyer does not. Conditional on a car being offered at a fixed price, suppose its two possible qualities are equally likely. The buyer must make the same purchase decision at both nodes of the information set below; a perfectly informative inspection splits this set and permits a quality-contingent decision. The two desirable outcomes---buying a good car and declining a lemon---require different actions at the original information set, so their conflict graph is $K_2$ and~\Cref{thm:fractional-formula} gives a value of perfect information of $\chi_f(K_2)=2$. More broadly, our framework isolates the value of inspection, certification, or screening in such environments.

\begin{center}
\begin{tikzpicture}[
    x=1cm,y=1cm,scale=0.9,transform shape,font=\small,
    chance/.style={circle,draw,inner sep=2pt},
    buyer/.style={circle,draw,fill=white,inner sep=2pt},
    infoset/.style={draw=blue!70,dashed,rounded corners,inner xsep=3pt,inner ysep=1.5pt},
    outcome/.style={circle,draw,inner sep=2pt},
    desirable/.style={outcome,fill=blue!12},
    terminal/.style={rectangle,rounded corners=2pt}
]
\foreach \shift/\prefix/\title in {0/u/Before inspection,5/r/After inspection} {
    \begin{scope}[xshift=\shift cm]
    \node at (0,0.55) {\title};
    \node[chance] (\prefix-c) at (0,0) {$c$};
    \node[buyer] (\prefix-g) at (-1.1,-1.1) {$b$};
    \node[buyer] (\prefix-l) at (1.1,-1.1) {$b$};
    \draw (\prefix-c) -- node[above left,font=\scriptsize] {good: $\frac12$} (\prefix-g);
    \draw (\prefix-c) -- node[above right,font=\scriptsize] {lemon: $\frac12$} (\prefix-l);
    \node[desirable,terminal] (\prefix-gb) at (-1.8,-2.5) {$1$};
    \node[outcome,terminal] (\prefix-gd) at (-0.4,-2.5) {$0$};
    \node[outcome,terminal] (\prefix-lb) at (0.4,-2.5) {$0$};
    \node[desirable,terminal] (\prefix-ld) at (1.8,-2.5) {$1$};
    \draw (\prefix-g) -- node[pos=0.65,sloped,above,font=\scriptsize] {buy} (\prefix-gb);
    \draw (\prefix-g) -- node[pos=0.65,sloped,above,font=\scriptsize] {decline} (\prefix-gd);
    \draw (\prefix-l) -- node[pos=0.65,sloped,above,font=\scriptsize] {buy} (\prefix-lb);
    \draw (\prefix-l) -- node[pos=0.65,sloped,above,font=\scriptsize] {decline} (\prefix-ld);
    \end{scope}
}
\node[infoset,fit=(u-g)(u-l)] {};
\node[infoset,fit=(r-g)] {};
\node[infoset,fit=(r-l)] {};
\node at (0,-3.05) {Optimal payoff: $\frac12$};
\node at (5,-3.05) {Optimal payoff: $1$};
\begin{scope}[xshift=9.2cm]
    \node[align=center] at (0,0.55) {Conflict graph $K_2$};
    \node[desirable,label={[font=\scriptsize]above:good, buy}] (v-g) at (0,-0.65) {};
    \node[desirable,label={[font=\scriptsize]below:lemon, decline}] (v-l) at (0,-1.65) {};
    \draw (v-g) -- (v-l);
    \node at (0,-3.05) {$\chi_f(K_2)=2$};
\end{scope}
\end{tikzpicture}
\end{center}

\subsection{Related work}
\label{sec:related-work}

From an algorithmic point of view, \citet{Shachter99:Efficient} considered the problem of computing the value of information based on \emph{influence diagrams}~\citep{Howard05:Influence}. Compared to the extensive form, an influence diagram is an alternative formalism for capturing decision making problems under uncertainty: it is a compact \emph{graphical model} that encodes probabilistic dependencies among random variables capturing uncertainty, rather than an explicit tree with information sets. In a related direction, \citet{Krause05:Near} considered (nonmyopically) selecting the most informative subset of variables in a graphical model, a problem tied to several practical applications, such as sensor network deployment. Prior work typically evaluates observations relative to fixed payoffs, whereas we measure the largest multiplicative gain over all payoff environments. From a complexity point of view, our results connect to research on \emph{bilevel optimization} and \emph{interdiction}; we refer to~\citet{Grune25:Completeness,Grune26:Complexity} for recent progress in such problems and further discussion.

Closer to our work is the recent framework of~\citet{Berker25:Value}, which introduced the value of recall. Their main focus was on bounding the value of recall in single-player decision problems, but they also considered settings with multiple players---games. For example, they observed that more recall \emph{may hurt} in two-player games; this counterintuitive behavior cannot arise in single-player decision problems. \citet{Alqithami26:Dynamic} operationalized value of recall in homophily-based network models so as to improve resilience under adversarial disruptions. Relatedly, \citet{Levy25:Information} study information design for AI proxies whose memories can be erased between periods, establishing a revelation principle. Following~\citet{Berker25:Value}, we define the value of information as a \emph{ratio}. This adheres to a long tradition in computer science in general, and algorithmic game theory in particular; for example, this is the case in notions such as the price of anarchy~\citep{Koutsoupias99:Worst}, the price of stability~\citep{Anshelevich08:Price}, the mediation value~\citep{Value08:Value}, and the value of commitment~\citep{Letchford14:Value}.

\paragraph{Broader context in economics.} Two additional foundational works concerning the role of information in economics are Blackwell's comparison of experiments~\citep{Blackwell53:Equivalent} and Hayek's influential account of prices as vessels of dispersed knowledge~\citep{Hayek45:Use}. Other early work continued viewing information through its effect on decision quality~\citep{Feltham68:Value}, while subsequent papers revealed subtle behavior, including the nonconcavity phenomenon of~\citet{Radner84:Nonconcavity} and, more recently, conditions under which concavity can be recovered~\citep{Whitmeyer25:Concavity}. Other work compares information structures or studies their value in different economic environments such as correlated equilibria in games~\citep{Gossner00:Comparison}, stochastic information structures~\citep{Azrieli08:Value}, production economies~\citep{Eckwert01:Value}, and teams with strategic information sharing~\citep{Gershkov16:How}. Related work in economics studies how retaining or deleting past records affects market outcomes. \citet{Liberman18:Equilibrium} examine a large-scale deletion of default records in Chile, finding that it redistributed access to credit and reduced aggregate borrowing. \citet{Kovbasyuk24:Memory} show that appropriately deleting records in a dynamic market can sustain trade when retaining unlimited histories would lead to market breakdown.

\section{Decision problems and information refinements}
\label{sec:model}

Before we proceed with our results, we provide some background. \Cref{subsec:extensive-form-decision-problems} begins by describing extensive-form decision problems and the notion of  perfect recall.  \Cref{subsec:refinements} defines the perfect-information and perfect-recall refinements.  \Cref{subsec:value-ratio} introduces a unifying value-ratio objective.

\subsection{Extensive-form decision problems}

\label{subsec:extensive-form-decision-problems}
A \emph{finite extensive-form decision problem} $\Gamma$ consists of a finite rooted tree~\cite{Kuhn53:Extensive,FudenbergTirole91:Game,OsborneRubinstein94:Course}. Its \emph{terminal nodes} (or leaves) are denoted by $L$. Every nonterminal node is either a \emph{chance node} or a \emph{player node}.  At each chance node $c$, chance chooses an outgoing action $a$ with rational probability $p_c(a)$, and the probabilities over the outgoing actions of $c$ sum to one. Player nodes are partitioned into \emph{information sets} $I\in\mathcal{I}(\Gamma)$; intuitively, an information set captures the set of nodes that the decision-maker cannot distinguish between. All nodes in a common information set have the same available \emph{action set}, denoted $A_I$, which bijectively maps to the immediate descendants of each node in $I$.

We write $H$ for the set of player nodes.  A \emph{pure strategy} $\pistrat$ of $\Gamma$ chooses one action $\pi_I\in A_I$ at every information set $I$ of $\Gamma$.  The finite set of pure strategies is denoted by $\Pi\defeq\Pi(\Gamma)$. A \emph{behavioral strategy} $\bstrat$ chooses a probability distribution over $A_I$ at every information set $I$.

For a terminal leaf $\ell\in L$, $p_c(\ell)$ denotes the product of chance probabilities on the root-to-$\ell$ path, excluding player choices.  We can eliminate all terminal leaves with $p_c(\ell)=0$; this does not change any expected payoff, so throughout the paper every leaf has $p_c(\ell) > 0$. We next introduce a property regarding the information structure of a decision problem. 

\begin{definition}[No absentmindedness]
\label{def:no-absentmindedness}
The decision problem $\Gamma$ has \emph{no absentmindedness} if every root-to-leaf path visits each information set at most once.
\end{definition}

No absentmindedness rules out a single play path entering the same information set twice, but it still allows the player to forget which information set or action was encountered earlier. Absentmindedness is a particular manifestation of imperfect recall. Throughout this paper, our blanket assumption is that the problem is free of absentmindedness.

\begin{definition}[Perfect recall]
\label{def:perfect-recall}
The decision problem $\Gamma$ has \emph{perfect recall} if any two player nodes in the same information set have the same sequence of the player's own previously visited information sets and own previously chosen actions.  Equivalently, information sets never merge histories that differ in the player's own past observations or actions.
\end{definition}

Perfect recall implies no absentmindedness. \Cref{fig:efg-concepts} illustrates these different concepts.

\paragraph{Payoff vectors and optimal strategies.} So far, we have not specified the underlying payoff structure. For a pure strategy $\pistrat\in\Pi$, let $\boldsymbol{x}^\pistrat\in\mathbb{R}_{\ge 0}^L$ denote the induced distribution over leaves:
\[
        x^\pistrat_\ell
        \defeq
        p_c(\ell)\,
        \mathbbm{1}\{\pistrat\text{ chooses every player action on the path to }\ell\}.
\]

Whenever $x^\pistrat_\ell>0$, we say $\pistrat$ is \emph{compatible} with $\ell$. A payoff vector $\boldsymbol{v}\in\mathbb{R}_{\ge 0}^L$ assigns a nonnegative payoff $v_\ell$ to each leaf, representing how much the decision maker values reaching outcome $\ell$. Importantly, we do \emph{not} take $\boldsymbol{v}$ to be part of the decision-problem description, allowing us to study the same decision problem $\Gamma$ under various payoff structures. For a given payoff vector $\boldsymbol{v}$, we define
\[
        V_{\Gamma}(\boldsymbol{v})
        \defeq
        \max_{\pistrat\in\Pi(\Gamma)}
        \sum_{\ell\in L} x^\pistrat_\ell v_\ell
\]
as the maximum attainable expected payoff; without absentmindedness, the maximum over behavioral strategies is always attained in pure strategies. For more background on imperfect recall and absentmindedness, we refer to~\citet{OsborneRubinstein94:Course} and~\citet{PiccioneRubinstein97:Interpretation}.

\begin{figure}[!t]
\centering
\begin{tikzpicture}[
    x=1cm,y=1cm,
    >=Latex,
    panel/.style={font=\small},
    chance/.style={circle,draw,minimum size=6mm,inner sep=0pt},
    player/.style={circle,draw,minimum size=6mm,inner sep=0pt},
    leaf/.style={rectangle,rounded corners,draw,minimum width=7mm,minimum height=4.5mm,inner sep=1pt},
    infoset/.style={rounded corners,dashed,draw=blue!70},
    edge/.style={->}
]

\begin{scope}[panel,shift={(-3.15,0)}]
\node[chance] (pi-r) at (0,0) {$c$};
\node[player] (pi-h1) at (-1.05,-1.15) {$h_1$};
\node[player] (pi-h2) at (1.05,-1.15) {$h_2$};
\draw[edge] (pi-r) -- node[left] {$p$} (pi-h1);
\draw[edge] (pi-r) -- node[right] {$1-p$} (pi-h2);
\node[leaf] (pi-l1) at (-1.65,-2.35) {$\ell_1$};
\node[leaf] (pi-l2) at (-0.45,-2.35) {$\ell_2$};
\node[leaf] (pi-l3) at (0.45,-2.35) {$\ell_3$};
\node[leaf] (pi-l4) at (1.65,-2.35) {$\ell_4$};
\draw[edge] (pi-h1) -- node[left] {$\textsf{L}$} (pi-l1);
\draw[edge] (pi-h1) -- node[right] {$\textsf{R}$} (pi-l2);
\draw[edge] (pi-h2) -- node[left] {$\textsf{L}$} (pi-l3);
\draw[edge] (pi-h2) -- node[right] {$\textsf{R}$} (pi-l4);
\node[align=center,font=\small] at (0,-3.05) {perfect\\information};
\end{scope}

\begin{scope}[panel,shift={(3.15,0)}]
\node[chance] (pr-r) at (0,0) {$c$};
\node[player] (pr-h1) at (-1.05,-1.15) {$h_1$};
\node[player] (pr-h2) at (1.05,-1.15) {$h_2$};
\draw[edge] (pr-r) -- node[pos=.35,left] {$p$} (pr-h1);
\draw[edge] (pr-r) -- node[pos=.35,right] {$1-p$} (pr-h2);
\node[leaf] (pr-l1) at (-1.65,-2.35) {$\ell_1$};
\node[leaf] (pr-l2) at (-0.45,-2.35) {$\ell_2$};
\node[leaf] (pr-l3) at (0.45,-2.35) {$\ell_3$};
\node[leaf] (pr-l4) at (1.65,-2.35) {$\ell_4$};
\draw[edge] (pr-h1) -- node[left] {$\textsf{L}$} (pr-l1);
\draw[edge] (pr-h1) -- node[right] {$\textsf{R}$} (pr-l2);
\draw[edge] (pr-h2) -- node[left] {$\textsf{L}$} (pr-l3);
\draw[edge] (pr-h2) -- node[right] {$\textsf{R}$} (pr-l4);
\node[infoset,fit=(pr-h1)(pr-h2),inner xsep=1mm,inner ysep=1mm] {};
\node[blue!70,font=\small,anchor=west] at (1.75,-1.15) {$I$};
\node[align=center,font=\small] at (0,-3.05) {imperfect information\\with perfect recall};
\end{scope}

\begin{scope}[panel,shift={(-3.15,-4.35)}]
\node[player] (ir-r) at (0,0) {$h_0$};
\node[player] (ir-h1) at (-1.05,-1.15) {$h_1$};
\node[player] (ir-h2) at (1.05,-1.15) {$h_2$};
\draw[edge] (ir-r) -- node[pos=.35,left] {$\textsf{L}_1$} (ir-h1);
\draw[edge] (ir-r) -- node[pos=.35,right] {$\textsf{R}_1$} (ir-h2);
\node[leaf] (ir-l1) at (-1.65,-2.35) {$\ell_1$};
\node[leaf] (ir-l2) at (-0.45,-2.35) {$\ell_2$};
\node[leaf] (ir-l3) at (0.45,-2.35) {$\ell_3$};
\node[leaf] (ir-l4) at (1.65,-2.35) {$\ell_4$};
\draw[edge] (ir-h1) -- node[left] {$\textsf{L}_2$} (ir-l1);
\draw[edge] (ir-h1) -- node[right] {$\textsf{R}_2$} (ir-l2);
\draw[edge] (ir-h2) -- node[left] {$\textsf{L}_2$} (ir-l3);
\draw[edge] (ir-h2) -- node[right] {$\textsf{R}_2$} (ir-l4);
\node[infoset,fit=(ir-h1)(ir-h2),inner xsep=1mm,inner ysep=1mm] {};
\node[blue!70,font=\small,anchor=west] at (1.75,-1.15) {$I$};
\node[align=center,font=\small] at (0,-3.05) {imperfect\\recall};
\end{scope}

\begin{scope}[panel,shift={(3.15,-4.35)}]
\node[player] (am-h1) at (0,0) {$h_1$};
\node[leaf] (am-l1) at (-1.25,-1.35) {$\ell_1$};
\node[player] (am-h2) at (1.05,-1.15) {$h_2$};
\draw[edge] (am-h1) -- node[left] {$\textsf{L}$} (am-l1);
\draw[edge] (am-h1) -- node[right] {$\textsf{R}$} (am-h2);
\node[leaf] (am-l2) at (0.45,-2.35) {$\ell_2$};
\node[leaf] (am-l3) at (1.65,-2.35) {$\ell_3$};
\draw[edge] (am-h2) -- node[left] {$\textsf{L}$} (am-l2);
\draw[edge] (am-h2) -- node[right] {$\textsf{R}$} (am-l3);
\node[infoset,fit=(am-h1)(am-h2),inner xsep=1mm,inner ysep=1mm] {};
\node[blue!70,font=\small,anchor=west] at (1.75,-0.57) {$I$};
\node[align=center,font=\small] at (0,-3.05) {absentmindedness};
\end{scope}

\end{tikzpicture}
\caption{Four information structures in extensive-form decision problems.  Outgoing arrow labels indicate action labels (for player nodes) or probabilities (for chance nodes). The dashed boxes indicate non-singleton information sets.}
\label{fig:efg-concepts}
\end{figure}
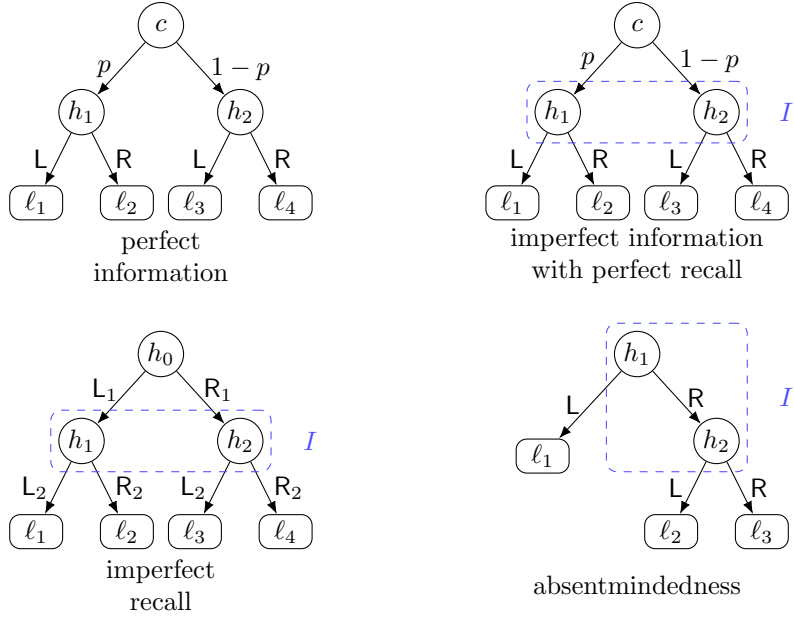

\subsection{Perfect-information and perfect-recall refinements}
\label{subsec:refinements}

We now introduce the notion of an \emph{(information) refinement} of a decision problem.

\begin{definition}[Refinement]
\label{def:refinement}
A \emph{refinement} of $\Gamma$ is a new decision problem $\widehat\Gamma$ obtained by splitting information sets of $\Gamma$ while leaving the tree, terminal leaves, action labels, and chance probabilities unchanged.  Thus, each refined information set is contained in an original information set and inherits its action set.
\end{definition}

Intuitively, $\widehat\Gamma$ provides the decision maker with more information than $\Gamma$, as they can now distinguish between more nodes. We now introduce two specific refinements central to our work.

\begin{definition}[Perfect-information refinement]
\label{def:pi-refinement}
For any decision problem $\Gamma$, its \emph{perfect-information refinement} $\GammaPI$ is the refinement in which every player node is placed in a singleton information set.
\end{definition}

\begin{definition}[Coarsest perfect-recall refinement]
\label{def:pr-refinement}
When $\Gamma$ has no absentmindedness, $\operatorname{seq}(h)$ denotes, for a player node $h$, the sequence of the player's own previously visited information sets and own actions on the root-to-$h$ path, measured in the original decision problem $\Gamma$.  The \emph{coarsest perfect-recall refinement} $\GammaPR$ partitions any original information set $I$ into new information sets defined by the equivalence relation
\[
        h\sim h'\quad\Longleftrightarrow\quad
        \operatorname{seq}(h)=\operatorname{seq}(h'),
        \quad h,h'\in I.
\]
\end{definition}
Equivalently, $\GammaPR$ distinguishes exactly the information that is required to restore perfect recall, and no more. In particular, \Cref{def:pr-refinement} is the \emph{coarsest refinement} used in the value-of-recall model of~\citet{Berker25:Value}: any other perfect recall refinement of $\Gamma$ is also a refinement of $\GammaPR$. If $\Gamma$ already has perfect recall, then $\GammaPR=\Gamma$.

\subsection{Value ratios}
\label{subsec:value-ratio}

Let $\widehat\Gamma$ be any refinement of $\Gamma$ in the sense of \Cref{def:refinement}. The \emph{value ratio} of refining $\Gamma$ to $\widehat\Gamma$ is defined as the supremum gain (in terms of terminal payoffs) over all possible payoff environments.

\begin{definition}[Value ratio]
\label{def:common-ratio}
The \emph{value ratio} of the refinement $\widehat\Gamma/\Gamma$ is
\[
        \VR(\widehat\Gamma/\Gamma)
        \defeq
        \sup_{\boldsymbol{v}\in\mathbb{R}_{\ge 0}^L,\ \boldsymbol{v}\ne 0}
        \frac{V_{\widehat\Gamma}(\boldsymbol{v})}{V_\Gamma(\boldsymbol{v})}.
\]
The \emph{value of perfect information} is $\operatorname{VoPI}(\Gamma)\defeq\VR(\GammaPI/\Gamma)$,  where $\Gamma$ is assumed to have perfect recall.  The \emph{value of recall} is $\operatorname{VoR}(\Gamma)\defeq\VR(\GammaPR/\Gamma)$, where $\Gamma$ is assumed to have no absentmindedness.
\end{definition}

Since we assume $p_c(\ell) > 0$ for all $\ell \in L$ and every root-to-leaf path is compatible with some pure strategy (under no absentmindedness), $V_\Gamma(\boldsymbol{v})>0$ for every nonzero $\boldsymbol{v}\ge 0$. Thus, the ratios are well-defined. Splitting information sets can only enlarge the set of pure strategies, so $V_{\widehat\Gamma}(\boldsymbol{v})\ge V_\Gamma(\boldsymbol{v})$ and every ratio is at least one.\footnote{Interestingly, this is not necessarily the case in settings with multiple players~\citep{Berker25:Value}.} For reference, our main notation is gathered in \Cref{tab:model-notation}.

\begin{table}[!htbp]
\caption{Notation used for the model and value ratio.}
\centering
\small
\renewcommand{\arraystretch}{1.15}
\begin{tabular}{@{}>{\raggedright\arraybackslash}p{0.27\linewidth}p{0.65\linewidth}@{}}
\toprule
\textbf{Notation} & \textbf{Meaning} \\
\midrule
\multicolumn{2}{@{}l}{\textit{Decision problems}} \\
\addlinespace[0.15em]
$\Gamma$ & Original finite extensive-form decision problem. \\
$L$ & Set of terminal leaves. \\
$H$ & Set of player nodes. \\
$I$ & Information set; the collection of information sets is $\mathcal{I}(\Gamma)$. \\
$A_I$ & Common action set at information set $I$. \\
$p_c(\ell)$ & Product of chance probabilities on the root-to-$\ell$ path.\\
$\Sigma(\Gamma)$ & Player sequences in $\Gamma$, including the empty sequence $\varnothing$. \\
$\sigma_I$ & Player sequence at information set $I$. \\
\midrule
\multicolumn{2}{@{}l}{\textit{Strategies}} \\
\addlinespace[0.15em]
$\Pi(\Gamma)$ & Set of pure strategies of decision problem $\Gamma$. \\
$\pistrat$ & Pure strategy of the original decision problem. \\
$\hatpistrat$ & Pure strategy of a refinement. \\
$\bstrat$ & Behavioral strategy. \\
\midrule
\multicolumn{2}{@{}l}{\textit{Leaf distributions}} \\
\addlinespace[0.15em]
$\boldsymbol{x}^\pistrat$, $\widehat{\boldsymbol{x}}^\hatpistrat$ & Leaf distributions induced by $\pistrat$ and $\hatpistrat$, respectively. \\
$\mathcal{X}_{\Gamma}$, $\mathcal{X}_{\widehat\Gamma}$ & Convex hulls of original and refined leaf-distribution vectors. \\
\midrule
\multicolumn{2}{@{}l}{\textit{Refinements and values}} \\
\addlinespace[0.15em]
$\GammaPI$ & Perfect-information refinement of $\Gamma$. \\
$\GammaPR$ & Coarsest perfect-recall refinement of $\Gamma$. \\
$\widehat\Gamma$ & Generic refinement of $\Gamma$ obtained by splitting information sets. \\
$V_{\Gamma'}(\boldsymbol{v})$ & Maximum (pure-strategy) expected payoff in $\Gamma'$ under payoff vector $\boldsymbol{v}$. \\
$\VR(\widehat\Gamma/\Gamma)$ & Value ratio of the refinement $\widehat\Gamma$ over $\Gamma$. \\
\bottomrule
\end{tabular}
\label{tab:model-notation}
\end{table}

\clearpage
\section{A unifying characterization through fractional coloring}
\label{sec:geometry}

The key result in this section is a characterization of the value ratio from \Cref{def:common-ratio} as the maximum fractional chromatic number of a certain family of graphs (\Cref{thm:fractional-formula}). This establishes a graph-theoretic interpretation not just of the value of perfect information or of perfect recall, but of any information refinement per~\Cref{def:refinement}.

In what follows, we fix a pure strategy $\hatpistrat\in\Pi(\widehat\Gamma)$ of the refinement. $L(\hatpistrat)\defeq\{\ell\in L:\widehat{x}^\hatpistrat_\ell>0\}$ denotes the set of leaves that $\hatpistrat$ reaches with nonzero probability. For a leaf $\ell\in L(\hatpistrat)$ and an information set $I \in\mathcal{I}(\Gamma)$ of the original game, $a_\ell(I)$ denotes the action taken at the visit to $I$ on the root-to-$\ell$ path, when that visit exists.\footnote{Since $\Gamma$ has no absentmindedness, no $I$ can be visited multiple times on the path to $\ell$, so $a_\ell(I)$ is well defined.}

\begin{definition}[Conflict graph]
\label{def:conflict-graph}
The \emph{conflict graph} $G_\hatpistrat$ has vertex set $L(\hatpistrat)$. Two distinct leaves $\ell,\ell'\in L(\hatpistrat)$ are adjacent if there exists an original information set $I$ such that both root-to-leaf paths visit $I$ and $a_\ell(I)\ne a_{\ell'}(I)$.
\end{definition}

That is, $\ell$ and $\ell'$ are adjacent if no pure strategy of the original decision problem is compatible with both leaves. The next lemma points out the key role of independent sets in the conflict graph. 

\begin{restatable}[Independent sets are compatible leaves]{lemma}{rstlemmaLemIndependentCompatible}
\label{lem:independent-compatible}
A set $S\subseteq L(\hatpistrat)$ is independent in $G_\hatpistrat$ if and only if there exists a pure original strategy $\pistrat\in\Pi$ compatible with every leaf in $S$.
\end{restatable}

(All proofs are deferred to \Cref{app:proofs}.) In this context, the \emph{fractional chromatic number} (\emph{e.g.}~\citep{ScheinermanUllman13:Fractional,Lovasz75:Ratio,LundYannakakis94:Hardness}) of a graph $G=(V,E)$ can be defined as
\begin{equation}
    \label{eq:dual-color}
    \max\left\{\sum_{v\in V}w_v:
        w_v\ge0,
        \sum_{v\in S}w_v\le1
        \text{ for every independent }S\right\}.
\end{equation}
While this is a linear program, it has exponentially many constraints. Indeed, deciding whether $\chi_f(G) > \tau$ is \coNP-complete~\citep{Grotschel81:Ellipsoid,ScheinermanUllman13:Fractional}. We now state the main characterization of this section.
\begin{restatable}{theorem}{rsttheoremThmFractionalFormula}
\label{thm:fractional-formula}
If $\Gamma$ has no absentmindedness,
$\VR(\widehat\Gamma/\Gamma) = \max_{\hatpistrat\in\Pi(\widehat\Gamma)}\chi_f(G_\hatpistrat)$.
\end{restatable}

Fractional coloring appears because, after rescaling payoffs by chance probabilities, the fixed refined strategy $\hatpistrat$ collects the total weight of all reached leaves (the weight of all other leaves can be set to zero WLOG), whereas any pure strategy in the original decision problem can collect only the weight of a compatible subset. By~\Cref{lem:independent-compatible}, these compatible subsets are exactly the independent sets of $G_\hatpistrat$. The maximum multiplicative discrepancy between those two quantities is precisely captured by~\eqref{eq:dual-color}. This characterization immediately gives us several observations about $\VR(\widehat\Gamma/\Gamma)$.

\begin{restatable}{corollary}{rstcorollaryCorCommonFractionalConsequences}
\label{cor:common-fractional-consequences}
For every refinement $\widehat\Gamma$ of a decision problem $\Gamma$ with no absentmindedness, 
\begin{enumerate}[label=(\roman*)]
\item $\VR(\widehat\Gamma/\Gamma)$ is rational,
\item the supremum in \Cref{def:common-ratio} is attained, and
\item $1\le\VR(\widehat\Gamma/\Gamma)\le |L|$.
\end{enumerate}
\end{restatable}

Having established our unifying characterization for the value of information of \emph{any} information refinement, we next employ \Cref{thm:fractional-formula} to investigate the computational complexity of two natural refinements: the \emph{value of perfect information} and the \emph{value of recall}.

\section{Value of perfect information}
\label{sec:voi}

We first examine the case where $\Gamma$ has perfect recall and $\widehat\Gamma=\GammaPI$ is the perfect-information refinement thereof. This section begins by proving that, under perfect recall, the fixed-strategy conflict graphs are \emph{cographs}, yielding a more fine-grained characterization of $\operatorname{VoPI}$ in terms of the clique number (\Cref{thm:voi-integer}). \Cref{subsec:voi-np-hardness} then proves \NP-completeness by a reduction from \textsc{MaximumCoverage}. Finally, \Cref{sec:FPT} provides a fixed-parameter tractability result.

\subsection{Cographs and an integer formula}
\label{subsec:voi-cographs}

A cograph is a graph obtained from single vertices by repeated \emph{disjoint union} and \emph{join} operations.  The join $G\vee G'$ of vertex-disjoint graphs $G$ and $G'$ is obtained from their disjoint union by adding all edges between $V(G)$ and $V(G')$. We use induction to show the following structural lemma.

\begin{restatable}[Cograph structure]{lemma}{rstlemmaLemCograph}
\label{lem:cograph}
If $\Gamma$ has perfect recall, then for every pure perfect-information strategy $\hatpistrat\in\Pi(\GammaPI)$, the conflict graph $G_\hatpistrat$ is a cograph.
\end{restatable}

Intuitively, different actions at the same information set create joins in the conflict graph, while chance-separated continuations after the same action create disjoint unions; perfect recall ensures that two subsequent information-set components cannot later merge and create cross conflicts.

For a graph $G$, we let $\omega(G)$ denote its clique number, that is, the size of the largest clique in $G$. We point out the following well-known fact.

\begin{restatable}[Fractional chromatic number of cographs~\citep{ScheinermanUllman13:Fractional}]{lemma}{rstlemmaLemCographFractional}
\label{lem:cograph-fractional}
If $G$ is a cograph, then $\chi_f(G)=\omega(G)$.
\end{restatable}

As a result, we obtain a more fine-grained characterization for the value of perfect information, establishing as a byproduct that it is always an integer. As we shall see, this is not always the case for the value of recall (\Cref{thm:graph-embedding}).

\begin{restatable}{corollary}{rsttheoremThmVoiInteger}
\label{thm:voi-integer}
For every perfect-recall decision problem $\Gamma$, $\operatorname{VoPI}(\Gamma) = \max_{\hatpistrat\in\Pi(\GammaPI)}\omega(G_\hatpistrat)$. In particular, $\operatorname{VoPI}(\Gamma)$ is an integer between $1$ and $\left\lceil\frac{|L|}{3}\right\rceil$.
\end{restatable}

We get a tighter upper bound than the trivial bound of $|L|$, as edges in the conflict graph $G_\hatpistrat$ indicate the existence of other leaves that $\hatpistrat$ does \emph{not} reach; we formalize this in \Cref{prop:unaccounted-leaves}.

\subsection{\NP-completeness of \textsc{VoPI}: reduction from \textsc{MaximumCoverage}}
\label{subsec:voi-np-hardness}

Next, we characterize the complexity of computing the value of perfect information. We examine \textsc{VoPI}-\textsc{Threshold}, the natural decision version of the problem: given a perfect-recall decision problem $\Gamma$ (without specifying terminal payoffs) and a rational threshold $\tau > 0$, decide whether $\operatorname{VoPI}(\Gamma)\ge \tau$. A strategy $\hatpistrat$ in the refined game yields a polynomial-time certificate for $\operatorname{VoPI}(\Gamma)\ge \tau$.

\begin{restatable}{proposition}{rstpropositionPropVoiInNp}
\label{prop:voi-in-np}
\textsc{VoPI}-\textsc{Threshold} is in \NP.
\end{restatable}

We next sketch our approach for showing \NP-hardness; the proof can be found in~\Cref{appendix:voi}. We reduce from \textsc{MaximumCoverage}. An instance consists of a universe $U=\{1,\ldots,n\}$, sets $S_1,\ldots,S_m\subseteq U$, an integer $k$, and a threshold $\tau$.  The question is whether $k$ sets cover at least $\tau$ elements.  We may assume that every element belongs to at least one set.

For an element $e\in U$, we write $M(e) \defeq \{j:e\in S_j\}$ and $d_e\defeq|M(e)|$. We construct a decision problem $\Gamma_{\mathcal{C}}$ from the coverage instance $\mathcal{C}$ as follows (\Cref{fig:coverage-reduction}).

\begin{enumerate}[label=\arabic*.]
\item Chance chooses an element $e\in U$ with probability $1/n$.  The resulting player node is $h_e$.
\item All nodes $h_e$ lie in one information set $I$.  The available first-stage actions are $a\in[k]$.
\item After action $a$, chance chooses $j\in M(e)$ with probability $1/d_e$.
\item The resulting player node is $h_{e,a,j}$.  For each pair $(a,j)$, all nodes $\{h_{e,a,j}:e\in S_j\}$ lie in one information set $J_{a,j}$.
\item At $J_{a,j}$ the player chooses an action $b\in U$.  Play then terminates at a leaf $\ell(e,a,j,b)$.
\end{enumerate}

For $j\in M(e)$, the joint chance probability of reaching $h_{e,a,j}$ after selecting $a$ as the first-stage action is $p_{e,j}\defeq\frac{1}{n d_e}$. The size of the decision problem is polynomial in $n$, $m$, $k$, and $\sum_{j=1}^m|S_j|$. Moreover, it is easy to see that $\Gamma_{\mathcal{C}}$ has perfect recall (\Cref{lem:coverage-decision-problem-pr}).

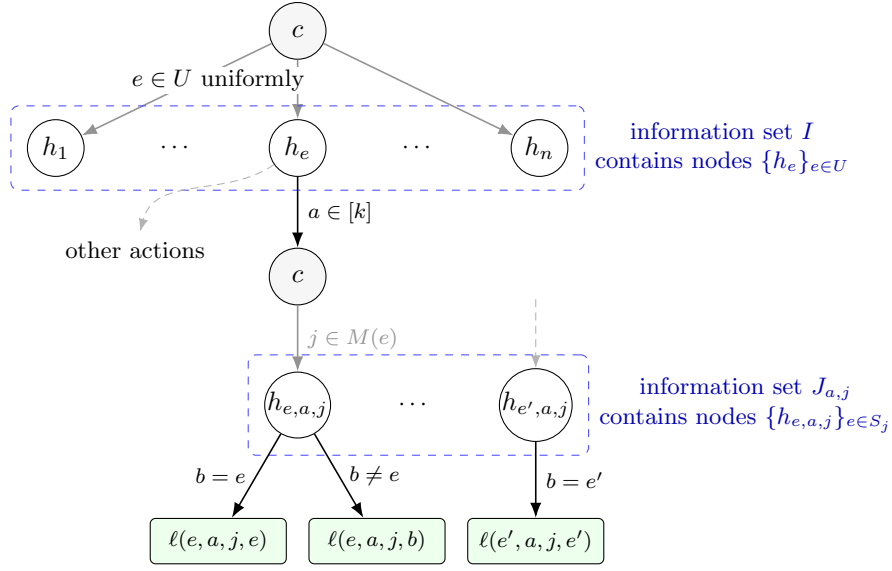
\begin{figure}[t]
\centering
\begin{tikzpicture}[
    x=1cm,y=1cm,
    >=Latex,
    chance/.style={circle,draw,minimum size=7.5mm,inner sep=0pt,font=\small,fill=gray!7},
    player/.style={circle,draw,minimum size=7.5mm,inner sep=0pt,font=\small,fill=white},
    leaf/.style={rectangle,draw,rounded corners=2pt,minimum width=18mm,minimum height=6mm,inner sep=2pt,font=\scriptsize,fill=green!7},
    infoset/.style={draw=blue!70,dashed,rounded corners=3pt,inner xsep=6pt,inner ysep=5pt},
    infofont/.style={font=\footnotesize,blue!70!black,align=center},
    edge/.style={->,semithick},
    chanceedge/.style={->,semithick,gray!85},
    note/.style={font=\footnotesize,align=center},
    edgelabel/.style={font=\footnotesize,align=center,fill=white,inner sep=1pt},
    faint/.style={gray!60,densely dashed}
]
\node[chance] (root) at (0,0) {$c$};

\node[player] (x1) at (-3.2,-1.55) {$h_1$};
\node[player] (xe) at (0,-1.55) {$h_e$};
\node[player] (xn) at (3.2,-1.55) {$h_n$};
\node[note] at (-1.6,-1.55) {$\cdots$};
\node[note] at (1.6,-1.55) {$\cdots$};
\draw[chanceedge] (root) -- (x1);
\draw[chanceedge] (root) -- (xe);
\draw[chanceedge] (root) -- (xn);
\node[edgelabel] at (-1.05,-0.65) {$e\in U$ uniformly};
\node[infoset,fit=(x1)(xe)(xn)] (Ibox) {};
\node[infofont,anchor=west] at (Ibox.east) {information set $I$\\contains nodes $\{h_e\}_{e \in U}$};

\node[chance] (cea) at (0,-3.25) {$c$};
\node[note] (othera) at (-2.15,-2.9) {other actions};
\draw[edge] (xe) -- node[right,font=\scriptsize] {$a\in[k]$} (cea);
\draw[faint,->] (xe) to[out=-145,in=75] (othera);

\node[player] (yeaj) at (0,-4.95) {$h_{e,a,j}$};
\node[player] (yeprimeaj) at (3.15,-4.95) {$h_{e',a,j}$};
\node[note] at (1.55,-4.95) {$\cdots$};
\draw[chanceedge] (cea) -- node[right,font=\scriptsize] {$j\in M(e)$} (yeaj);
\draw[faint,->] (3.15,-3.55) -- (yeprimeaj);
\node[infoset,fit=(yeaj)(yeprimeaj)] (Jbox) {};
\node[infofont,anchor=west] at (Jbox.east) {information set $J_{a,j}$\\contains nodes $\{h_{e,a,j}\}_{e \in S_j}$};

\node[leaf] (succ) at (-1.05,-6.75) {$\ell(e,a,j,e)$};
\node[leaf] (fail) at (1.05,-6.75) {$\ell(e,a,j,b)$};
\node[leaf] (succp) at (3.15,-6.75) {$\ell(e',a,j,e')$};
\draw[edge] (yeaj) -- node[left,font=\scriptsize] {$b=e$} (succ);
\draw[edge] (yeaj) -- node[right,font=\scriptsize] {$b\ne e$} (fail);
\draw[edge] (yeprimeaj) -- node[right,font=\scriptsize] {$b=e'$} (succp);

\end{tikzpicture}
\caption{Extensive-form decision problem used in the reduction from \textsc{MaximumCoverage} to the value of perfect information. Chance first selects an element $e$, the player chooses an action $a$ at the shared information set $I$, chance then selects an incident set index $j\in M(e)$, and the player finally picks an element $b$ at the shared information set $J_{a,j}$.}
\label{fig:coverage-reduction}
\end{figure}

In this construction, the $k$ first-stage actions serve as coverage slots. A perfect-information strategy assigns each element to one of the $k$ chosen sets. On the other hand, the original strategy cannot condition this assignment on the element. The upshot is that, when optimizing the payoffs in this construction, the gap reveals the optimal value of~\textsc{MaximumCoverage}. Indeed, if $\OPT\defeq\max_{j_1,\ldots,j_k\in[m]} \left|\bigcup_{a=1}^k S_{j_a}\right|$ is the optimum coverage value (with repetitions allowed), we claim that the value of perfect information in $\Gamma_{\mathcal{C}}$ matches the optimal value of $\mathcal{C}$.

\begin{restatable}{theorem}{rsttheoremThmVoiCoverageExact}
\label{thm:voi-coverage-exact}
For the decision problem $\Gamma_{\mathcal{C}}$ constructed above, $\operatorname{VoPI}(\Gamma_{\mathcal{C}})=\OPT$.
\end{restatable}

We thus arrive at the following \NP-completeness result.

\begin{restatable}{theorem}{rsttheoremThmVoiNpComplete}
\label{thm:voi-np-complete}
\textsc{VoPI}-\textsc{Threshold} is \NP-complete.
\end{restatable}

\paragraph{Approximation consequence.} The reduction above is exact, in that it preserves the optimum value. Thus, the standard inapproximability for \textsc{MaximumCoverage} due to~\citet{Feige98:Threshold} translates to inapproximability for the optimization version of the value of perfect information.

\begin{corollary}[Approximation hardness]
\label{cor:voi-approx-hardness}
For every constant $\epsilon>0$, unless $\P=\NP$, the optimization version of the value of perfect information admits no polynomial-time $(1-1/e+\epsilon)$-approximation.
\end{corollary}

\subsection{Fixed-parameter tractability}
\label{sec:FPT}

The foregoing hardness results notwithstanding, a natural question is whether the value of perfect information can be computed efficiently in more structured classes of decision problems. We positively answer this question for perfect-recall settings in which the largest information-set size is bounded, capturing decision problems where the amount of imperfect information is circumscribed.

\begin{restatable}[Fixed-parameter tractability]{theorem}{voidp}
\label{thm:bounded-width-dp}
There is an algorithm that computes $\operatorname{VoPI}(\Gamma)$ in time
\[
        \operatorname{poly}(|\Gamma|)
        \sum_{I\in\mathcal I(\Gamma)}
        (1+|A_I|)^{|I|}.
\]
In particular, if $\max_{I \in \mathcal{I}(\Gamma)} |I| = O(1)$, the algorithm is polynomial in the description size $|\Gamma|$.
\end{restatable}
Moreover, the algorithm is \emph{fixed-parameter tractable} in terms of the largest information-set size and the largest action set. Our algorithm is based on dynamic programming and relies on the fact that the conflict graph is a cograph (\Cref{lem:cograph}) in order to set up a suitable recursion, as we formalize in~\Cref{sec:voi-algorithms}.

\section{Value of recall}
\label{sec:vor}

We next turn to the setting where $\Gamma$ may have imperfect recall (but no absentmindedness) and $\widehat\Gamma=\GammaPR$ is the coarsest perfect-recall refinement. We write $\operatorname{VoR}(\Gamma)\defeq\VR(\GammaPR/\Gamma)$; \Cref{thm:fractional-formula} establishes that $\operatorname{VoR}(\Gamma) = \max_{\hatpistrat\in\Pi(\GammaPR)}\chi_f(G_\hatpistrat)$, where $G_\hatpistrat$ is the conflict graph (\Cref{def:conflict-graph}).

\subsection{\texorpdfstring{Membership in $\class{\Sigma^p_2}$}{Membership in Sigma-2-p}}
\label{subsec:vor-membership}

We examine~\textsc{VoR}-\textsc{Threshold}, the associated decision version of the problem.

\begin{definition}[\textsc{VoR}-\textsc{Threshold}]
\label{def:vor-threshold}
An instance consists of an explicitly given finite extensive-form decision problem $\Gamma$ with rational chance probabilities and no absentmindedness, together with a rational threshold $\tau \in \mathbb
Q_{>0}$. The question is whether $\operatorname{VoR}(\Gamma) \geq \tau$.
\end{definition}

Computing the value of recall can only be harder than the value of perfect information: as noted by \citet[Proposition~4]{Berker25:Value}, adding dummy player nodes to any decision problem $\Gamma$ is sufficient to construct a strategically equivalent problem $\Gamma'$ such that giving perfect recall in $\Gamma'$ amounts to giving perfect information in $\Gamma$. Thus, $\operatorname{VoR}(\Gamma')=\operatorname{VoPI}(\Gamma)$ and the complexity result from the preceding section (\Cref{thm:voi-np-complete}) implies \NP-hardness for the analogous problem of \textsc{VoR}-\textsc{Threshold} (\Cref{def:vor-threshold}).

The proof of \NP~membership for \textsc{VoPI}-\textsc{Threshold} (\Cref{prop:voi-in-np}), on the other hand, does not carry over: even if we were given the refined strategy $\hatpistrat$ that attains the maximum in~\Cref{thm:fractional-formula}, ascertaining whether the fractional coloring (of an arbitrary graph) is at least a certain threshold is \coNP-hard. This tension between the inner and outer optimizations suggests that $\class{\Sigma^p_2}$ is the suitable complexity class for this problem.
\begin{restatable}[Membership]{proposition}{rstpropositionPropVorInSigmatwo}
\label{prop:vor-in-sigma2}
\textsc{VoR}-\textsc{Threshold} belongs to $\class{\Sigma^p_2}$.
\end{restatable}

As a warm-up, we next establish a particularly crisp reduction from the fractional coloring problem, which we then upgrade to $\class{\Sigma^p_2}$-hardness in~\Cref{sec:vor-sigma2-hardness} using several additional ingredients.

\subsection{Fractional coloring embedding}
\label{subsec:vor-np-hardness}

Let $G=(V,E)$ be a simple graph. We fix an arbitrary orientation of every edge and a total order $\prec$ on $E$. For each vertex $v$, we list its incident edges in the order inherited from $\prec$: $e_1(v)\prec e_2(v)\prec\cdots\prec e_{d(v)}(v)$, where $d(v)\defeq\deg_G(v)$.

\begin{definition}[Graph embedding construction]
\label{def:graph-embedding-construction}
The decision problem $\Gamma_G$ is defined as follows.
\begin{enumerate}[label=\arabic*.]
\item At the root, chance chooses $v\in V$ uniformly.
\item On the branch selected by $v$, the player first visits a singleton \emph{marker} information set $M_v$ having one dummy action $\bot$.
\item If $d(v)=0$, the dummy action leads directly to a terminal success leaf $\ell_v$.  If $d(v)>0$, it leads to a player node $h_{v,1}$.  For $i=1,\ldots,d(v)$, the node $h_{v,i}$ belongs to the original information set $I_{e_i(v)}$, whose action set is $\{0,1\}$.
\item At $h_{v,i}$, the continuing action is $0$ if $v$ is the tail endpoint of the oriented edge $e_i(v)$ and $1$ if $v$ is the head endpoint.  This action leads to $h_{v,i+1}$ when $i<d(v)$ and to $\ell_v$ when $i=d(v)$; the other action leads immediately to a private terminal exit leaf $t_{v,i}$.
\end{enumerate}
\end{definition}

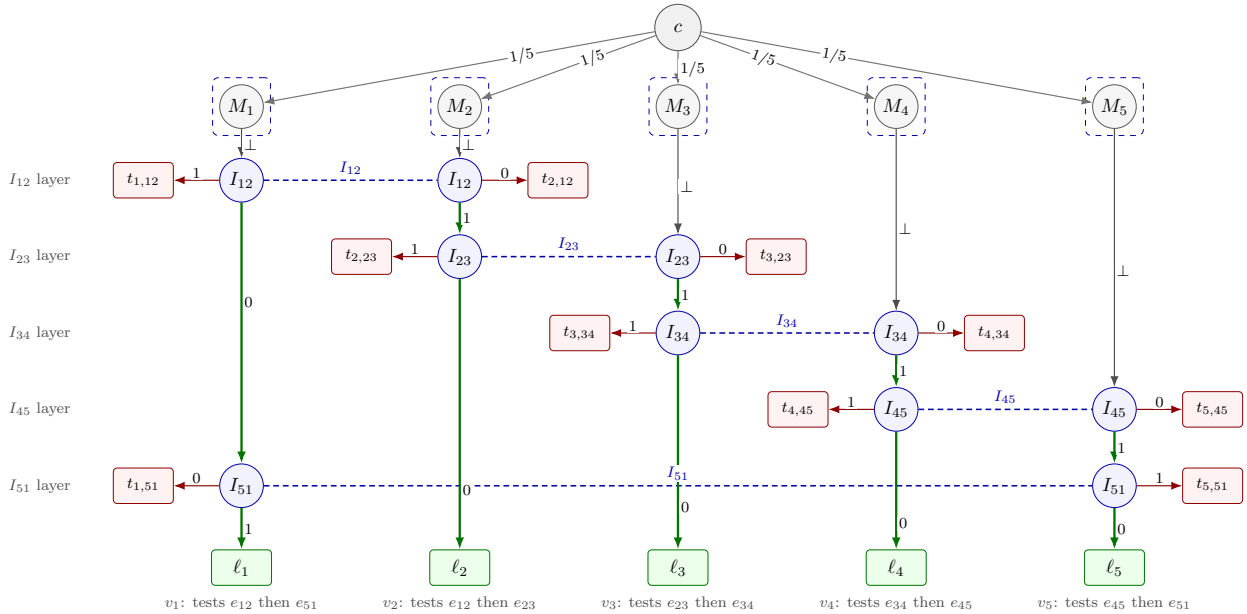
\begin{figure}[!htbp]
    \centering
    \resizebox{\linewidth}{!}{\begin{tikzpicture}[
    x=1cm,y=1cm,
    >=Latex,
    chance/.style={circle,draw=black!70,fill=gray!12,minimum size=8.5mm,inner sep=0pt,font=\small},
    marker/.style={circle,draw=black!65,fill=gray!8,minimum size=8mm,inner sep=1.4pt,font=\small},
    test/.style={circle,draw=blue!65!black,fill=blue!5,minimum size=8mm,inner sep=0pt,font=\small},
    success/.style={rectangle,rounded corners=2pt,draw=green!45!black,fill=green!8,minimum width=11mm,minimum height=6.5mm,inner sep=1.2pt,font=\small},
    exit/.style={rectangle,rounded corners=2pt,draw=red!55!black,fill=red!5,minimum width=11mm,minimum height=6.5mm,inner sep=1.2pt,font=\scriptsize},
    cont/.style={draw=green!45!black,very thick,-{Latex[length=2.2mm]}},
    quit/.style={draw=red!55!black,semithick,-{Latex[length=2mm]}},
    chanceedge/.style={draw=black!55,semithick,-{Latex[length=2mm]}},
    branch/.style={draw=black!70,semithick,-{Latex[length=2mm]}},
    infoline/.style={draw=blue!65!black,densely dashed,thick},
    infoset/.style={draw=blue!65!black,dashed,rounded corners=3pt,inner sep=3.5pt},
    smalllabel/.style={font=\scriptsize,inner sep=1pt,fill=white},
    note/.style={font=\scriptsize,align=center,text=black!70}
]


\node[chance] (root) at (0,0) {$c$};
\coordinate (X1) at (-8,0);
\coordinate (X2) at (-4,0);
\coordinate (X3) at ( 0,0);
\coordinate (X4) at ( 4,0);
\coordinate (X5) at ( 8,0);

\node[marker] (M1) at (-8,-1.45) {$M_1$};
\node[marker] (M2) at (-4,-1.45) {$M_2$};
\node[marker] (M3) at ( 0,-1.45) {$M_3$};
\node[marker] (M4) at ( 4,-1.45) {$M_4$};
\node[marker] (M5) at ( 8,-1.45) {$M_5$};

\draw[chanceedge] (root) -- node[smalllabel,pos=.34,sloped] {$1/5$} (M1);
\draw[chanceedge] (root) -- node[smalllabel,pos=.36,sloped] {$1/5$} (M2);
\draw[chanceedge] (root) -- node[smalllabel,right,pos=.55] {$1/5$} (M3);
\draw[chanceedge] (root) -- node[smalllabel,pos=.36,sloped] {$1/5$} (M4);
\draw[chanceedge] (root) -- node[smalllabel,pos=.34,sloped] {$1/5$} (M5);

\node[note,anchor=east] at (-11.0,-2.8) {$I_{12}$ layer};
\node[note,anchor=east] at (-11.0,-4.2) {$I_{23}$ layer};
\node[note,anchor=east] at (-11.0,-5.6) {$I_{34}$ layer};
\node[note,anchor=east] at (-11.0,-7.0) {$I_{45}$ layer};
\node[note,anchor=east] at (-11.0,-8.4) {$I_{51}$ layer};

\node[test] (I12a) at (-8,-2.8) {$I_{12}$};  
\node[test] (I12b) at (-4,-2.8) {$I_{12}$};  
\node[test] (I23a) at (-4,-4.2) {$I_{23}$};  
\node[test] (I23b) at ( 0,-4.2) {$I_{23}$};  
\node[test] (I34a) at ( 0,-5.6) {$I_{34}$};  
\node[test] (I34b) at ( 4,-5.6) {$I_{34}$};  
\node[test] (I45a) at ( 4,-7.0) {$I_{45}$};  
\node[test] (I45b) at ( 8,-7.0) {$I_{45}$};  
\node[test] (I51b) at (-8,-8.4) {$I_{51}$};  
\node[test] (I51a) at ( 8,-8.4) {$I_{51}$};  

\node[success] (ell1) at (-8,-9.9) {$\ell_1$};
\node[success] (ell2) at (-4,-9.9) {$\ell_2$};
\node[success] (ell3) at ( 0,-9.9) {$\ell_3$};
\node[success] (ell4) at ( 4,-9.9) {$\ell_4$};
\node[success] (ell5) at ( 8,-9.9) {$\ell_5$};

\node[exit] (T112) at (-9.8,-2.8) {$t_{1,12}$};
\node[exit] (T212) at (-2.2,-2.8) {$t_{2,12}$};
\node[exit] (T223) at (-5.8,-4.2) {$t_{2,23}$};
\node[exit] (T323) at ( 1.8,-4.2) {$t_{3,23}$};
\node[exit] (T334) at (-1.8,-5.6) {$t_{3,34}$};
\node[exit] (T434) at ( 5.8,-5.6) {$t_{4,34}$};
\node[exit] (T445) at ( 2.2,-7.0) {$t_{4,45}$};
\node[exit] (T545) at ( 9.8,-7.0) {$t_{5,45}$};
\node[exit] (T151) at (-9.8,-8.4) {$t_{1,51}$};
\node[exit] (T551) at ( 9.8,-8.4) {$t_{5,51}$};

\draw[branch] (M1) -- node[smalllabel,right,pos=.55] {$\bot$} (I12a);
\draw[branch] (M2) -- node[smalllabel,right,pos=.55] {$\bot$} (I12b);
\draw[branch] (M3) -- node[smalllabel,right,pos=.55] {$\bot$} (I23b);
\draw[branch] (M4) -- node[smalllabel,right,pos=.55] {$\bot$} (I34b);
\draw[branch] (M5) -- node[smalllabel,right,pos=.55] {$\bot$} (I45b);

\draw[cont] (I12a) -- node[smalllabel,right,pos=.38] {$0$} (I51b);
\draw[quit] (I12a) -- node[smalllabel,above] {$1$} (T112);
\draw[cont] (I51b) -- node[smalllabel,right] {$1$} (ell1);
\draw[quit] (I51b) -- node[smalllabel,above] {$0$} (T151);

\draw[cont] (I12b) -- node[smalllabel,right] {$1$} (I23a);
\draw[quit] (I12b) -- node[smalllabel,above] {$0$} (T212);
\draw[cont] (I23a) -- node[smalllabel,right,pos=.78] {$0$} (ell2);
\draw[quit] (I23a) -- node[smalllabel,above] {$1$} (T223);

\draw[cont] (I23b) -- node[smalllabel,right] {$1$} (I34a);
\draw[quit] (I23b) -- node[smalllabel,above] {$0$} (T323);
\draw[cont] (I34a) -- node[smalllabel,right,pos=.78] {$0$} (ell3);
\draw[quit] (I34a) -- node[smalllabel,above] {$1$} (T334);

\draw[cont] (I34b) -- node[smalllabel,right] {$1$} (I45a);
\draw[quit] (I34b) -- node[smalllabel,above] {$0$} (T434);
\draw[cont] (I45a) -- node[smalllabel,right,pos=.78] {$0$} (ell4);
\draw[quit] (I45a) -- node[smalllabel,above] {$1$} (T445);

\draw[cont] (I45b) -- node[smalllabel,right] {$1$} (I51a);
\draw[quit] (I45b) -- node[smalllabel,above] {$0$} (T545);
\draw[cont] (I51a) -- node[smalllabel,right] {$0$} (ell5);
\draw[quit] (I51a) -- node[smalllabel,above] {$1$} (T551);

\begin{scope}[on background layer]
  \draw[infoline] (I12a) -- (I12b);
  \draw[infoline] (I23a) -- (I23b);
  \draw[infoline] (I34a) -- (I34b);
  \draw[infoline] (I45a) -- (I45b);
  \draw[infoline] (I51b) -- (I51a);

  \node[infoset,fit=(M1)] {};
  \node[infoset,fit=(M2)] {};
  \node[infoset,fit=(M3)] {};
  \node[infoset,fit=(M4)] {};
  \node[infoset,fit=(M5)] {};
\end{scope}

\node[smalllabel,text=blue!65!black] at ($.5*(I12a)+.5*(I12b)+(0,.23)$) {$I_{12}$};
\node[smalllabel,text=blue!65!black] at ($.5*(I23a)+.5*(I23b)+(0,.23)$) {$I_{23}$};
\node[smalllabel,text=blue!65!black] at ($.5*(I34a)+.5*(I34b)+(0,.23)$) {$I_{34}$};
\node[smalllabel,text=blue!65!black] at ($.5*(I45a)+.5*(I45b)+(0,.23)$) {$I_{45}$};
\node[smalllabel,text=blue!65!black] at ($.5*(I51a)+.5*(I51b)+(0,.23)$) {$I_{51}$};

\node[note] at (-8,-10.55) {$v_1$: tests $e_{12}$ then $e_{51}$};
\node[note] at (-4,-10.55) {$v_2$: tests $e_{12}$ then $e_{23}$};
\node[note] at ( 0,-10.55) {$v_3$: tests $e_{23}$ then $e_{34}$};
\node[note] at ( 4,-10.55) {$v_4$: tests $e_{34}$ then $e_{45}$};
\node[note] at ( 8,-10.55) {$v_5$: tests $e_{45}$ then $e_{51}$};

\end{tikzpicture}}
    \caption{The decision problem $\Gamma_{C_5}$ from \Cref{def:graph-embedding-construction} for the oriented cycle $v_1\to v_2\to v_3\to v_4\to v_5\to v_1$ with edge order $e_{12}\prec e_{23}\prec e_{34}\prec e_{45}\prec e_{51}$.  Each original edge information set occupies one horizontal layer; nonincident layers are omitted on each vertex branch.  The continuing actions realize the cycle conflicts under the all-continuing refined strategy. Green edges indicate continuing actions, red edges indicate exit leaves, blue dashed horizontal segments indicate original information sets, and singleton dashed boxes indicate marker information sets $M_i$.}
    \label{fig:full-vor-fractional}
\end{figure}

The key claim, established in~\Cref{app:vor-graph-embedding}, is that the value of recall of $\Gamma_G$ is equal to the fractional chromatic number of $G$.

\begin{restatable}{theorem}{NPfractional}
\label{thm:graph-embedding}
For every graph $G$, the decision problem $\Gamma_G$ of \Cref{def:graph-embedding-construction} is free of absentmindedness, has size $O(|V|+|E|)$, and satisfies $\operatorname{VoR}(\Gamma_G)=\chi_f(G)$.
\end{restatable}

\Cref{fig:full-vor-fractional} illustrates an example of this reduction when $G$ is the $5$-cycle. In fact, this example also shows that \emph{the value of recall can be fractional}, unlike the value of perfect information; we elaborate on this in~\Cref{subsec:c5-example} below. First, combining with the inapproximability of~\citet{Zuckerman06:Linear} for the fractional coloring number, we point out the following consequence.

\begin{corollary}
    \label{cor:VoR-inapprox}
    For any constant $\epsilon \in (0, 1/2)$, approximating VoR within a factor of $|L|^{1/2 - \epsilon }$ is \NP-hard.
\end{corollary}

\FloatBarrier
\subsubsection{\texorpdfstring{Fractional example of VoR: embedding the $5$-cycle}{Fractional example: the 5-cycle}}
\label{subsec:c5-example}

We now formalize the fact that, unlike the value of perfect information, the value of recall need not be an integer. In particular, such an example already arises from our previous construction applied to the cycle $C_5$.

\begin{figure}[!htbp]
\centering
\begin{tikzpicture}[
    x=1cm,y=1cm,
    >=Latex,
    vertex/.style={circle,draw,minimum size=9mm,inner sep=0pt,fill=white,font=\small},
    certbox/.style={draw=gray!55,line width=0.2pt,rounded corners=2pt,anchor=west,align=left,inner xsep=7pt,inner ysep=6pt,font=\small},
    title/.style={font=\bfseries\small,align=center},
    note/.style={font=\footnotesize,align=center,text width=40mm},
    result/.style={inner xsep=9pt,inner ysep=5pt,font=\bfseries\small}
]

\node[title] at (0,2.05) {$C_5$};
\foreach \i/\ang in {1/90,2/18,3/-54,4/-126,5/162} {
    \node[vertex] (v\i) at ({1.25*cos(\ang)},{1.25*sin(\ang)}) {$v_\i$};
}
\foreach \a/\b in {1/2,2/3,3/4,4/5,5/1} {
    \draw[semithick] (v\a) -- (v\b);
}
\node[result] at (0,-2.55) {$\chi_f(C_5)=\frac52$};

\node[certbox] (dual) at (5.5,1.55) {
    \textbf{Dual lower bound}\\[1mm]
    Set $w_{v_i}=\frac12$ for $i=1,\ldots,5$.\\
    For every $S\in\Ind(C_5)$, $\sum_{v\in S}w_v\le 1$.\\[1mm]
    Objective value: $\sum_{i=1}^5w_{v_i}=\frac52$.
};

\node[certbox] (primal) at (5.5,-1.75) {
    \textbf{Fractional coloring}\\[1mm]
    Use the five independent sets\\
    $S_1=\{v_1,v_3\}$,\quad $S_2=\{v_2,v_4\}$,\quad $S_3=\{v_3,v_5\}$,\\
    $S_4=\{v_4,v_1\}$,\quad $S_5=\{v_5,v_2\}$.\\[1mm]
    Set $z_{S_i}=\frac12 \text{ } \forall i$; each vertex is covered twice.
};

\draw[gray!55,line width=0.4pt]
    (3.55,0 |- dual.north) -- (3.55,0 |- primal.south);

\end{tikzpicture}
\caption{The fractional-coloring certificate behind \Cref{prop:vor-fractional-example}.  The uniform dual weights give $\chi_f(C_5)\ge 5/2$, while the five displayed independent sets with weight $1/2$ each give $\chi_f(C_5)\le 5/2$.  Applying the graph-embedding construction of \Cref{def:graph-embedding-construction} therefore yields an imperfect-recall decision problem $\Gamma_{C_5}$ with no absentmindedness and $\operatorname{VoR}(\Gamma_{C_5})=\chi_f(C_5)=5/2$.}
\label{fig:c5}
\end{figure}
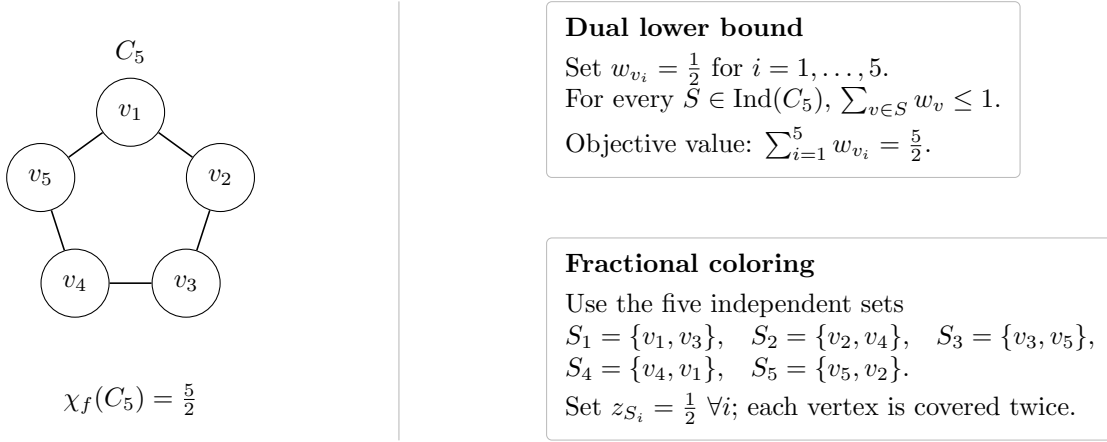

\begin{proposition}[Fractional value of recall]
\label{prop:vor-fractional-example}
There exists a decision problem $\Gamma$ with no absentmindedness such that $\operatorname{VoR}(\Gamma) = 5/2$. In particular, the value of recall need not be an integer.
\end{proposition}

\begin{proof}
The decision problem produced by \Cref{def:graph-embedding-construction} from the graph $C_5$ is $\Gamma=\Gamma_{C_5}$.  Then $\operatorname{VoR}(\Gamma)=\chi_f(C_5)$. It remains to compute $\chi_f(C_5)$. For the lower bound, we use the dual linear program for fractional chromatic number:
\[
        \chi_f(G)=
        \max\left\{\sum_{v\in V}w_v:
        w_v\ge0,
        \sum_{v\in S}w_v\le1
        \text{ for every independent }S\right\}.
\]
We assign weight $w_v=1/2$ to every vertex $v$ of $C_5$.  Every independent set of $C_5$ has size at most $2$, so every independent set has total weight at most $1$.  Thus, $\boldsymbol{w}$ is dual feasible, and
\[
        \chi_f(C_5)\ge \sum_{v\in V(C_5)} w_v = 5\cdot\frac12=\frac52.
\]
For the upper bound, we consider the five independent sets
\[
        S_1\defeq\{v_1,v_3\},\quad
        S_2\defeq\{v_2,v_4\},\quad
        S_3\defeq\{v_3,v_5\},\quad
        S_4\defeq\{v_4,v_1\},\quad
        S_5\defeq\{v_5,v_2\}.
\]
We assign weight $z_{S_i}=1/2$ to each $S_i$, and weight $0$ to every other independent set.  Each vertex belongs to exactly two of these five independent sets, so the covering constraints of the primal LP are satisfied:
\[
        \sum_{S\in\Ind(C_5):\, v\in S} z_S = 1
        \quad\text{for every }v\in V(C_5).
\]
Therefore, the total primal weight is
\[
        \sum_{S\in\Ind(C_5)} z_S = 5\cdot\frac12=\frac52.
\]
This shows that $\chi_f(C_5)\le 5/2$.  Combining the two bounds yields $\chi_f(C_5)=5/2$, and therefore $\operatorname{VoR}(\Gamma)=5/2$.
\end{proof}

We next build on this reduction to settle the complexity of the decision version of~$\operatorname{VoR}$.

\FloatBarrier
\subsection{\texorpdfstring{$\class{\Sigma^p_2}$-hardness of \textsc{VoR}}{Sigma-2-p hardness}}
\label{sec:vor-sigma2-hardness}

To prove $\class{\Sigma^p_2}$-hardness, we reduce directly from \textsc{ExistsForallDNFSAT}~\citep{Stockmeyer76:Polynomial}, with input formula $\exists \forall \Phi$. The first step is to reduce it to $3$-colorability under what we refer to as~\emph{optional edges}, which capture the existential assignment. In particular, the optional edges fix the existential variables in a $3$-colorability encoding of $\neg\Phi$; the resulting graph is $3$-colorable precisely when some universal assignment falsifies $\Phi$. Non-$3$-colorability therefore captures the universal condition.

\begin{restatable}{lemma}{rstlemmaLemOptionalEdgeThreecol}
\label{lem:optional-edge-3col}
Let $\Phi(X,Y)$ be a DNF formula with $X=\{x_1,\ldots,x_n\}$ and with no empty term. One can construct in polynomial time a graph $G_0=(W,E_0)$, a triangle with vertices $\mathsf T,\mathsf F,\mathsf B$, and, for every $i\in[n]$ and $b\in\{0,1\}$, a vertex $\lambda_i^{(b)}\in W$ and an optional edge $e_i^{(b)}=\{\lambda_i^{(b)},\mathsf T\}$ such that, for every $\boldsymbol{a}\in\{0,1\}^n$, the graph $G_{\boldsymbol{a}}\defeq G_0+\{e_i^{(a_i)}:i=1,\ldots,n\}$ is $3$-colorable if and only if there exists an assignment to $Y$ for which $\Phi(\boldsymbol{a},Y)=\mathsf{FALSE}$. Equivalently, $G_{\boldsymbol{a}}$ is not $3$-colorable if and only if $\forall Y\,\Phi(\boldsymbol{a},Y)$.
\end{restatable}

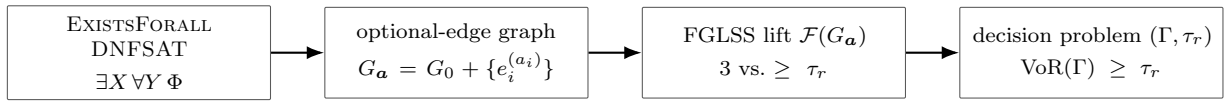
\begin{figure}[!htbp]
\centering
\begin{tikzpicture}[
    x=1cm,y=1cm,
    >=Latex,
    stage/.style={draw=black!65,rounded corners=0.5pt,align=center,inner xsep=4pt,inner ysep=4pt,font=\scriptsize,minimum width=25mm,minimum height=12mm,text width=32mm,fill=white},
    source/.style={stage},
    graph/.style={stage},
    lift/.style={stage},
    problem/.style={stage},
    decision/.style={stage},
    arrow/.style={->,thick},
    arlabel/.style={font=\scriptsize,align=center,fill=white,inner sep=1.2pt}
]

\node[source] (sat) at (0,0) {
    \textsc{ExistsForall}\\
    \textsc{DNFSAT}\\[1mm]
    $\exists X\,\forall Y\,\Phi$
};
\node[graph] (j) at (4.2,0) {
    optional-edge graph \\[1mm]
    $G_{\boldsymbol{a}}=G_0+\{e_i^{(a_i)}\}$
};
\node[lift] (k) at (8.4,0) {
    FGLSS lift $\mathcal{F}(G_{\boldsymbol{a}})$ \\[1mm]
    $3$ vs.\ $\ge \tau_r$
};
\node[problem] (g) at (12.6,0) {
    decision problem $(\Gamma,\tau_r)$\\[1mm]
    $\operatorname{VoR}(\Gamma)\ge \tau_r$
};

\draw[arrow] (sat) -- (j);
\draw[arrow] (j) -- (k);
\draw[arrow] (k) -- (g);

\end{tikzpicture}
\caption{Overview of the $\class{\Sigma^p_2}$-hardness reduction.}
\label{fig:sigma2-reduction-pipeline}
\end{figure}

To pass from coloring to fractional coloring, we use a conflict-graph construction following the approach of~\citet{Feige91:Approximating}, applied here to the constraints of proper 3-coloring. For a graph $G=(W,E)$, we define $\mathcal{F}(G)$ as follows. The vertex set is $V(\mathcal{F}(G))\defeq W\times\{0,1,2\}$, while the edges are of two kinds: i) $(w,c)$ is adjacent to $(w,c')$ whenever $c\ne c'$, and ii) $(w,c)$ is adjacent to $(w',c)$ whenever $\{w, w'\}\in E$.

\begin{restatable}[Fractional coloring gap]{lemma}{rstlemmaLemFglssGap}
\label{lem:fglss-gap}
For a graph $G=(W,E)$ with $r=|W|\ge2$, the following hold.
\begin{enumerate}[label=(\roman*)]
\item If $G$ is 3-colorable, then $\chi_f(\mathcal{F}(G))=3$.
\item If $G$ is not 3-colorable, then $\chi_f(\mathcal{F}(G))\ge \frac{3r}{r-1}
        =3+\frac{3}{r-1}$.
\end{enumerate}
\end{restatable}

For the graph $G_0$ produced above, we let $r=|W|$ and $\tau_r\defeq 3+\frac{3}{2(r-1)}$. $\tau_r$ is strictly between the two cases in \Cref{lem:fglss-gap}: if $G$ is 3-colorable then $\chi_f(\mathcal{F}(G))=3<\tau_r$, while if $G$ is not 3-colorable then $\chi_f(\mathcal{F}(G))\ge 3+3/(r-1)>\tau_r$.

The final step is to construct a suitable imperfect-recall decision problem. Compared to the reduction of~\Cref{subsec:vor-np-hardness}, the main addition is what we call \emph{selector information sets} $S_i$ whose actions implement the choice of the optional-edge vector $\boldsymbol{a}$; after those selector choices route the relevant branches, the remaining part of the reduction follows the graph embedding of~\Cref{def:graph-embedding-construction}.

\begin{restatable}{lemma}{vorreal}
\label{lem:vor-graph-realization}
For the constructed decision problem, $\max_{\hatpistrat\in\Pi(\GammaPR)}\chi_f(G_\hatpistrat) = \max_{\boldsymbol{a}\in\{0,1\}^n}\chi_f(\mathcal{F}_{\boldsymbol{a}})$.
\end{restatable}

We are thus able to characterize the complexity of the value of recall.

\begin{restatable}[$\class{\Sigma^p_2}$-completeness of value of recall]{theorem}{rsttheoremThmVorSigmatwoComplete}
\label{thm:vor-sigma2-complete}
\textsc{VoR}-\textsc{Threshold} is $\class{\Sigma^p_2}$-complete. This holds even for decision problems with a single initial chance node and no absentmindedness.
\end{restatable}

\section{Conclusion and future research}

We introduced a new general measure of the value of information in extensive-form decision problems, capturing the largest possible utility gain from any given \emph{information refinement}. We characterized the value of information through the maximum fractional chromatic number over an associated family of conflict graphs, giving rise to a unified combinatorial view of different information refinements. Using this characterization, we showed that computing the value of perfect information is \NP-complete, even in perfect-recall decision problems, while computing the value of recall is \class{\Sigma^p_2}-complete (without absentmindedness). Thus, different forms of information refinement have fundamentally different computational properties.

Our results suggest several directions for future work. First, it is worth examining the complexity of other natural information refinements. For example, what can be said about the value of perfect information of an \emph{imperfect-recall} decision problem? It would also be interesting to understand the complexity landscape in decision problems with absentmindedness. Perhaps the most important question is to explore the value of information (or, indeed, the \emph{cost} of information~\citep{Berker25:Value}) beyond single-player decision problems.

\section*{Acknowledgments}
R.E.B., E.T., and V.C.\ thank the Cooperative AI Foundation, Macroscopic Ventures  (formerly Polaris Ventures / the Center for Emerging Risk Research) and Jaan Tallinn’s donor-advised fund at Founders Pledge for financial support. R.E.B.\ and E.T.\ are also supported by the Cooperative AI PhD Fellowship. I.A. and T.S. are supported by NIH award A240108S001, the Vannevar Bush Faculty Fellowship ONR N00014-23-1-2876, and National Science Foundation grant RI-2312342. B.H.Z.\ was supported by ONR MURI grant N00014-25-1-2116 and ONR grant N00014-25-1-2296.

We are grateful to the anonymous reviewers at WINE for helpful suggestions that improved the presentation of this paper, as well as to Jason Li for helpful discussions on the connection to fractional coloring.

\section*{Statement on AI use}

We used GPT 5.5 Pro to assist with the development of some proofs. In particular, for~\Cref{thm:voi-coverage-exact}, we had explored similar reductions, but GPT found the construction that led to the proof. The proof of $\class{\Sigma^p_2}$-hardness in~\Cref{thm:vor-sigma2-complete} was found through a repeated interaction with GPT, starting from an approach we had been developing. We also used GPT for edits throughout the paper. The authors take full responsibility for the content of this paper.

 \bibliographystyle{plainnat}
\bibliography{refs}

\clearpage
\appendix
\crefalias{section}{appendix}
\crefalias{subsection}{appendix}
\crefalias{subsubsection}{appendix}

\section{Algorithms for the value of perfect information}
\label{sec:voi-algorithms}

The main result of this section is a dynamic programming algorithm for computing the value of perfect information in decision problems with bounded information-set size. First, we describe a simpler algorithm based on enumerating all perfect information strategies.

\subsection{Linear program in perfect-recall decision problems}
\label{subsec:fixed-strategy-lp}

To begin with, we point out that when the original decision problem $\Gamma$ has perfect recall, there is a compact linear program under a fixed refined strategy; this can be used for value-of-perfect-information algorithms.

The key idea is to use the \emph{sequence-form representation}. In particular, we denote by $\Sigma(\Gamma)$ the set of \emph{player sequences} in $\Gamma$, including the empty sequence $\varnothing$.  If an information set $I$ has player sequence $\sigma_I$ and action $a\in A_I$, we write $(\sigma_I,a)$ for the sequence obtained by appending $a$. We write $Q_\Gamma$ for the \emph{sequence-form (realization) polytope}~\cite{KollerEtAl96:Efficient,Romanovskii62:Reduction}:
\[
        Q_\Gamma\defeq
        \left\{\boldsymbol{q}\in\mathbb{R}_{\ge 0}^{\Sigma(\Gamma)}:
        q_\varnothing=1,
        \sum_{a\in A_I}q_{(\sigma_I,a)}
        =q_{\sigma_I}
        \text{ for every information set }I
        \right\}.
\]
For a terminal leaf $\ell$, $\sigma_\ell$ denotes the entire sequence on the path to $\ell$. For a fixed pure strategy $\hatpistrat\in\Pi(\widehat\Gamma)$, we write $L(\hatpistrat)\defeq\{\ell\in L:\widehat{x}^\hatpistrat_\ell>0\}$. Since all chance probabilities are assumed to be positive, $\widehat{x}^\hatpistrat_\ell=p_c(\ell)$ on $L(\hatpistrat)$.  We define $\lambda(\hatpistrat)$ as the optimum value of the linear program
\begin{align*}
        \text{maximize}\quad & t\\
        \text{subject to}\quad
        & \boldsymbol{q}\in Q_\Gamma,\\
        & q_{\sigma_\ell}\ge t
          &&\text{for every }\ell\in L(\hatpistrat).
\end{align*}
In particular, for a fixed $\hatpistrat\in\Pi(\widehat\Gamma)$, this can be solved in polynomial time. The key claim is provided below. (It is clear that $\lambda(\hatpistrat) > 0$.)

\begin{restatable}[Fixed-strategy sequence-form LP]{proposition}{rstpropositionPropFixedLp}
\label{prop:fixed-lp}
If $\Gamma$ has perfect recall, then for every pure refined strategy $\hatpistrat$, $\rho(\widehat{\boldsymbol{x}}^\hatpistrat) = 1/ \lambda(\hatpistrat)$, where $\rho$ is defined in~\Cref{prop:support}.
\end{restatable}

\begin{proof}
By \Cref{prop:downward-containment}, $\rho(\widehat{\boldsymbol{x}}^\hatpistrat)$ is the least $\alpha$ for which some $\boldsymbol{x}\in\mathcal{X}_\Gamma$ satisfies $\widehat{\boldsymbol{x}}^\hatpistrat\preceq\alpha\boldsymbol{x}$.  The polytope $\mathcal{X}_\Gamma$ is the set of leaf distributions induced by mixed strategies over pure strategies of $\Gamma$.  Since $\Gamma$ has perfect recall, mixed strategies are realization-equivalent to sequence-form realization plans: equivalently, $\boldsymbol{x}\in\mathcal{X}_\Gamma$ if and only if there is a vector $\boldsymbol{q}\in Q_\Gamma$ such that $x_\ell=p_c(\ell)q_{\sigma_\ell}$ for each $\ell \in L$. Therefore,
$\rho(\widehat{\boldsymbol{x}}^\hatpistrat)$ is the least $\alpha$ for which some such original sequence-form strategy $\boldsymbol{q}$ satisfies
\[
        \widehat{x}^\hatpistrat_\ell\le \alpha\,p_c(\ell) q_{\sigma_\ell}
\]
for each $\ell \in L$.
For $\ell\in L(\hatpistrat)$, this is
\[
        p_c(\ell)\le \alpha p_c(\ell) q_{\sigma_\ell},
\]
which, because $p_c(\ell)>0$, is equivalent to $q_{\sigma_\ell}\ge 1/\alpha$.  For leaves outside $L(\hatpistrat)$ the inequality is vacuous.  Thus, minimizing $\alpha$ is exactly the same as maximizing $t=1/\alpha$ subject to the linear program defining $\lambda(\hatpistrat)$.
\end{proof}

\paragraph{Enumeration algorithm.} By \Cref{prop:fixed-lp}, an algorithm can simply enumerate all pure perfect-information strategies $\hatpistrat\in\Pi(\GammaPI)$,and solve in each case a linear program with polynomial size. The overall running time is $\poly(|\Gamma|) \prod_{h\in H}|A_h|$, where we recall that $A_h$ is the action set available at player node $h$; in $\GammaPI$, each player node is a singleton information set, so this product is the number of pure strategies in the perfect-information refinement.

\subsection{Bounded information-set size}

We now describe an algorithm that is fixed-parameter tractable in terms of the size of the largest information set and the largest action set. For convenience, we add a new singleton information set $I_0$ above
the root, with one action leading to the original root; this does
not change $\operatorname{VoPI}(\Gamma)$.

For an information set $I$ and a subset $S\subseteq I$, we let $F(I,S)$ be the largest clique size obtainable among the leaves descending from $S$, optimizing over all perfect-information choices
at and below those nodes. Perfect recall implies that all nodes of $I$ have the same previous player sequence, so no additional history has to be incorporated in the state.

For $T\subseteq I$ and $a\in A_I$, suppose that action $a$ is taken
at every node of $T$. We then follow chance until either a terminal
leaf or the next player node is reached.  Let $M(a,T)$ be the largest of the following:
\begin{itemize}
\item $0$;
\item $1$, if a terminal leaf is reached before another player node;
\item $F(J,U)$, for each information set $J$ for which $U\subseteq J$
      is the nonempty set of first player nodes reached in $J$.
\end{itemize}
In particular, $M(a,\varnothing)=0$. A perfect-information choice at the nodes of $S$ is a map
$\phi:S\to A_I$. The recurrence we consider is
\begin{equation}
\label{eq:bounded-width-dp}
F(I,S)
=
\max_{\phi:S\to A_I}
\sum_{a\in A_I}
M\bigl(a,\phi^{-1}(a)\bigr).
\end{equation}

Our goal is to establish the following theorem.

\voidp*

\begin{proof}
We consider an information set $I$, $S\subseteq I$, and a map $\phi:S\to A_I$.  For each action
$a$, the set $\phi^{-1}(a)$ consists of the nodes at which $\phi$ chooses $a$.

Leaves descending from two different action classes are pairwise adjacent in the conflict graph: their paths visit the same information set $I$ but require different actions there. This means that the subgraphs generated by the different action classes are joined, and their clique sizes add.

Let us now consider one action $a$. No conflict is created at $I$ among leaves descending from $\phi^{-1}(a)$, since all of them take the same action. After that action, chance-only play separates these leaves into the following pieces: the terminal leaves reached before another
player decision, and, for each possible next information set $J$, the leaves descending from the first reached nodes $U\subseteq J$.

There are no conflict edges between distinct pieces. Terminal leaves have no later player decisions. Moreover, if leaves entering two
different next information sets $J$ and $J'$ later visited a common information set, the corresponding histories there would have different previous player sequences: one contains $J$ immediately after $I$, while the other contains $J'$. This would contradict perfect recall. Thus, these pieces form a disjoint union.

The terminal piece has clique number $1$ when it is nonempty, and the piece entering $J$ through $U$ has optimal clique number $F(J,U)$. A clique can therefore use only one such piece, and its best possible contribution from action $a$ is exactly $M(a,\phi^{-1}(a))$. It follows that, for the fixed map $\phi$, the largest obtainable clique has size $\sum_{a\in A_I}M\bigl(a,\phi^{-1}(a)\bigr)$.

Conversely, we now consider what happens when we choose for each action a component and a clique attaining the corresponding $M$-value. The perfect-information choices realizing these cliques can be combined because they occur at distinct player nodes, and their union is a clique because different action classes are joined. Maximizing over $\phi$ proves \eqref{eq:bounded-width-dp}.

Every information set appearing on the right-hand side lies strictly later in the player's history than $I$.  The recurrence is therefore acyclic and can be evaluated in reverse topological order.  At the dummy information set $I_0$, the state $F(I_0,I_0)$ covers the whole decision problem, so the clique characterization of the value of perfect information (\Cref{thm:voi-integer}) gives $F(I_0,I_0)=\operatorname{VoPI}(\Gamma)$.

It remains to bound the running time.  For a fixed information set $I$, the total number of maps considered over all subsets $S\subseteq I$ is
\[
        \sum_{S\subseteq I}|A_I|^{|S|}
        =
        (1+|A_I|)^{|I|}.
\]
The successor pieces used to evaluate $M(a,T)$ can be precomputed, for every action $a$ and subset $T\subseteq I$, by traversing chance edges until the next player node or terminal leaf. Summing over the information sets proves the running time bound. Finally, if $w=\max_I|I|$ and $A=\max_I|A_I|$, then the running time is $(1+A)^w\operatorname{poly}(|\Gamma|)$. Since $A\le |\Gamma|$, this is bounded by $|\Gamma|^{O(w)}$, and is hence polynomial when $w$ is fixed.
\end{proof}

\section{Proofs}
\label{app:proofs}

This section contains the proofs of the formal statements from the main body.

\subsection{Geometric formulation}

We first establish some basic lemmas concerning the geometric formulation of the value ratio. As a stepping stone, we provide a characterization of the value ratio as a downward containment in terms of the underlying leaf-distribution polytopes (\Cref{prop:downward-containment}).

For a pure strategy $\pistrat\in\Pi$, we recall that $\boldsymbol{x}^\pistrat\in\mathbb{R}_{\ge 0}^L$ denotes the induced distribution over the leaves. For a pure strategy of the refinement $\hatpistrat\in\Pi(\widehat\Gamma)$, $\widehat{\boldsymbol{x}}^\hatpistrat\in\mathbb{R}_{\ge 0}^L$ is defined analogously. Thus, $\boldsymbol{x}^\pistrat$ and $\widehat{\boldsymbol{x}}^\hatpistrat$ are both distributions over leaves. We write
\[
        \mathcal{X}_{\Gamma}\defeq\operatorname{conv}\{\boldsymbol{x}^\pistrat:\pistrat\in\Pi\}
        \text{ and }
        \mathcal{X}_{\widehat\Gamma}\defeq\operatorname{conv}\{\widehat{\boldsymbol{x}}^\hatpistrat:\hatpistrat\in\Pi(\widehat\Gamma)\}.
\]
For a compact convex set $\mathcal{K}\subset\mathbb{R}^L$, we write $h_{\mathcal{K}}(\boldsymbol{v})\defeq\max_{\boldsymbol{z}\in \mathcal{K}}\langle \boldsymbol{v},\boldsymbol{z}\rangle$ for its \emph{support function}. The lemma below follows directly from the definitions.

\begin{restatable}{lemma}{rstlemmaPropSupport}
\label{prop:support}
The value ratio satisfies
\[
        \VR(\widehat\Gamma/\Gamma)
        =
        \max_{\hatpistrat\in\Pi(\widehat\Gamma)}\rho(\widehat{\boldsymbol{x}}^\hatpistrat),
        \text{ where }
        \rho(\widehat{\boldsymbol{x}})\defeq\sup_{\boldsymbol{v}\ge 0,\ \boldsymbol{v}\ne 0}
        \frac{\langle \boldsymbol{v},\widehat{\boldsymbol{x}}\rangle}{h_{\mathcal{X}_{\Gamma}}(\boldsymbol{v})}.
\]
\end{restatable}

\begin{proof}
By the definition of $V_\Gamma$ and the leaf-distribution vectors, for every nonzero $\boldsymbol{v}\ge0$, 
\[
        V_\Gamma(\boldsymbol{v})=\max_{\pistrat\in\Pi}\langle \boldsymbol{v},\boldsymbol{x}^\pistrat\rangle=h_{\mathcal{X}_{\Gamma}}(\boldsymbol{v}).
\]
Using the finiteness of $\Pi(\widehat\Gamma)$ to exchange a finite maximum with a supremum gives
\[
        \VR(\widehat\Gamma/\Gamma)
        =
        \max_{\hatpistrat\in\Pi(\widehat\Gamma)}
        \sup_{\boldsymbol{v}\ge0,\ \boldsymbol{v}\ne0}
        \frac{\langle \boldsymbol{v},\widehat{\boldsymbol{x}}^\hatpistrat\rangle}
             {h_{\mathcal{X}_{\Gamma}}(\boldsymbol{v})},
\]
which is the claimed formula.
\end{proof}

For payoff vectors $\boldsymbol{a},\boldsymbol{b}\in\mathbb{R}^L$, we write $\boldsymbol{a}\preceq \boldsymbol{b}$ if $a_\ell\le b_\ell$ for every $\ell\in L$.  The \emph{downward closure} of the original leaf-distribution polytope is $\mathcal{D}_\Gamma\defeq\mathcal{X}_{\Gamma}-\mathbb{R}_{\ge 0}^L
        =\{\boldsymbol{x}-\boldsymbol{r}:\boldsymbol{x}\in\mathcal{X}_{\Gamma},
                 \boldsymbol{r}\in\mathbb{R}_{\ge 0}^L\}$. 

\begin{restatable}[Downward-containment characterization]{proposition}{rstpropositionPropDownwardContainment}
\label{prop:downward-containment}
For every $\widehat{\boldsymbol{x}}\in\mathcal{X}_{\widehat\Gamma}$,
\[
        \rho(\widehat{\boldsymbol{x}})=
        \min\{\alpha\ge0: \exists \boldsymbol{x}\in\mathcal{X}_{\Gamma}
             \text{ such that }\widehat{\boldsymbol{x}}\preceq \alpha \boldsymbol{x}\}.
\]
As a result,
\[
        \VR(\widehat\Gamma/\Gamma)
        =
        \min\{\alpha\ge0:
        \mathcal{X}_{\widehat\Gamma}\subseteq \alpha\mathcal{D}_\Gamma\}.
\]
\end{restatable} 

A point of $\mathcal{X}_{\Gamma}$ is an outcome distribution generated by randomized behavior in the original information structure, while a point of $\mathcal{X}_{\widehat\Gamma}$ is an outcome distribution attainable after refinement. The containment $\mathcal{X}_{\widehat\Gamma}\subseteq \alpha \mathcal{D}_\Gamma$ therefore says that every refined outcome distribution $\widehat{\boldsymbol{x}}$ can be coordinatewise dominated by $\alpha\boldsymbol{x}$ for some original randomized outcome distribution $\boldsymbol{x}$.  Equivalently, \Cref{prop:downward-containment} casts $\VR(\widehat\Gamma/\Gamma)$ as the smallest multiplicative factor by which (randomized) behavior available under the original information structure must be scaled so that it dominates every outcome distribution attainable after refinement.

\begin{proof}[Proof of~\Cref{prop:downward-containment}]
First assume that $\widehat{\boldsymbol{x}}\preceq \alpha \boldsymbol{x}$ for some $\boldsymbol{x}\in\mathcal{X}_{\Gamma}$.  Then for every $\boldsymbol{v}\ge0$,
\[
        \langle \boldsymbol{v},\widehat{\boldsymbol{x}}\rangle
        \le \alpha\langle \boldsymbol{v},\boldsymbol{x}\rangle
        \le \alpha h_{\mathcal{X}_{\Gamma}}(\boldsymbol{v}),
\]
so for every nonzero $\boldsymbol{v}\ge0$,
\[
        \frac{\langle \boldsymbol{v},\widehat{\boldsymbol{x}}\rangle}
             {h_{\mathcal{X}_{\Gamma}}(\boldsymbol{v})}
        \le \alpha.
\]
Taking the supremum over nonzero $\boldsymbol{v}\ge0$ in the definition of
\[
        \rho(\widehat{\boldsymbol{x}})
        =
        \sup_{\boldsymbol{v}\ge0,\ \boldsymbol{v}\ne0}
        \frac{\langle \boldsymbol{v},\widehat{\boldsymbol{x}}\rangle}
             {h_{\mathcal{X}_{\Gamma}}(\boldsymbol{v})}
\]
gives $\rho(\widehat{\boldsymbol{x}})\le\alpha$. Conversely, suppose $\rho(\widehat{\boldsymbol{x}})\le\alpha$.  By the same expression for $\rho$, for every nonzero $\boldsymbol{v}\ge0$,
\[
        \langle \boldsymbol{v},\widehat{\boldsymbol{x}}\rangle
        \le \alpha h_{\mathcal{X}_{\Gamma}}(\boldsymbol{v}),
\]
and the inequality is trivial for $\boldsymbol{v}=0$.  If $\widehat{\boldsymbol{x}}\notin\alpha(\mathcal{X}_{\Gamma}-\mathbb{R}_{\ge 0}^L)$, then the closed convex set $\alpha(\mathcal{X}_{\Gamma}-\mathbb{R}_{\ge 0}^L)$ can be strictly separated from $\widehat{\boldsymbol{x}}$.  Thus, there exists $\boldsymbol{y}\in\mathbb{R}^L$ such that
\[
        \langle \boldsymbol{y},\widehat{\boldsymbol{x}}\rangle>
        \sup_{\boldsymbol{z}\in\alpha(\mathcal{X}_{\Gamma}-\mathbb{R}_{\ge 0}^L)}\langle \boldsymbol{y},\boldsymbol{z}\rangle.
\]
The supremum on the right is finite only if $\boldsymbol{y}\ge0$. For such $\boldsymbol{y}$,
\[
        \sup_{\boldsymbol{z}\in\alpha(\mathcal{X}_{\Gamma}-\mathbb{R}_{\ge 0}^L)}\langle \boldsymbol{y},\boldsymbol{z}\rangle
        =\alpha h_{\mathcal{X}_{\Gamma}}(\boldsymbol{y}),
\]
contradicting the assumed inequality.  Hence $\widehat{\boldsymbol{x}}\in\alpha(\mathcal{X}_{\Gamma}-\mathbb{R}_{\ge 0}^L)$, which is equivalent to the existence of $\boldsymbol{x}\in\mathcal{X}_{\Gamma}$ with $\widehat{\boldsymbol{x}}\preceq \alpha \boldsymbol{x}$.  Taking the least admissible $\alpha$ proves
\[
        \rho(\widehat{\boldsymbol{x}})
        =
        \min\{\alpha\ge0:\widehat{\boldsymbol{x}}\in \alpha\mathcal{D}_\Gamma\}
        =
        \min\{\alpha\ge0:\exists\boldsymbol{x}\in\mathcal{X}_\Gamma
        \text{ such that }\widehat{\boldsymbol{x}}\preceq\alpha\boldsymbol{x}\}.
\]

It remains to pass from one refined outcome distribution to the whole refined polytope.  Since $\mathcal{X}_{\widehat\Gamma}$ is the convex hull of the finitely many pure refined leaf-distribution vectors and $\alpha\mathcal{D}_\Gamma$ is convex, the containment $\mathcal{X}_{\widehat\Gamma}\subseteq\alpha\mathcal{D}_\Gamma$ holds if and only if every $\widehat{\boldsymbol{x}}^\hatpistrat$ lies in $\alpha\mathcal{D}_\Gamma$.  Therefore, the least such $\alpha$ is $\max_{\hatpistrat\in\Pi(\widehat\Gamma)}
        \rho(\widehat{\boldsymbol{x}}^\hatpistrat)$, which equals $\VR(\widehat\Gamma/\Gamma)$ by \Cref{prop:support}.
\end{proof}

\subsection{Conflict graphs and fractional coloring}

\rstlemmaLemIndependentCompatible*
\begin{proof}
If $\pistrat$ is compatible with every leaf in $S$, then no two leaves in $S$ can require different actions at a common original information set, so $S$ is independent.

Conversely, if $S$ is independent, then for every information set $I$ visited by at least one leaf in $S$, all leaves in $S$ visiting $I$ require the same action; otherwise two such leaves would be adjacent.  We define $\pi_I$ to be this common action.  Information sets not visited by leaves in $S$ receive arbitrary legal actions.  Then $\pistrat$ is compatible with every leaf in $S$.
\end{proof}

\rsttheoremThmFractionalFormula*

Before we proceed with the proof, we recall the primal formulation of the fractional chromatic number:
\[
    \chi_f(G)\defeq
        \min\left\{\sum_{S\in\Ind(G)} z_S:
        z_S\ge0,
        \sum_{S\in\Ind(G):\, v\in S}z_S\ge1\text{ for every }v\in V\right\},
\]
where $\Ind(G)$ denotes the set of independent sets of $G$. The dual of this linear program reads
\[
    \max\left\{\sum_{v\in V}w_v:
        w_v\ge0,
        \sum_{v\in S}w_v\le1
        \text{ for every }S \in \Ind(G) \right\}.
\]

\begin{proof}[Proof of~\Cref{thm:fractional-formula}]

For a graph $G=(V,E)$ and nonnegative weights $\boldsymbol{w}\in\mathbb{R}_{\ge 0}^V$, the weighted independence number is $\alpha_{\boldsymbol{w}}(G)\defeq\max\left\{\sum_{v\in S}w_v : S \subseteq V\text{ independent}\right\}$. We will show that for every pure refined strategy $\hatpistrat$, $\rho(\widehat{\boldsymbol{x}}^\hatpistrat)=\chi_f(G_\hatpistrat)$. In the definition of $\rho(\widehat{\boldsymbol{x}}^\hatpistrat)$, payoffs outside $L(\hatpistrat)$ can be set to zero without decreasing the numerator and without increasing the denominator. We can thus assume $v_\ell=0$ for $\ell\notin L(\hatpistrat)$.  For $\ell\in L(\hatpistrat)$, we let $w_\ell\defeq p_c(\ell)v_\ell$. Since $p_c(\ell)>0$, arbitrary nonnegative weights $\boldsymbol{w}$ on $L(\hatpistrat)$ arise from suitable nonnegative payoffs $\boldsymbol{v}$. The refined strategy $\hatpistrat$ obtains payoff
\[
        \langle \boldsymbol{v},\widehat{\boldsymbol{x}}^\hatpistrat\rangle=
        \sum_{\ell\in L(\hatpistrat)}p_c(\ell)v_\ell
        =\sum_{\ell\in L(\hatpistrat)}w_\ell.
\]
 For an original pure strategy $\pistrat$, let $C(\pistrat)\subseteq L(\hatpistrat)$ be the leaves compatible with $\pistrat$.  Every such $C(\pistrat)$ is independent in $G_\hatpistrat$, so
\[
        h_{\mathcal{X}_{\Gamma}}(\boldsymbol{v})
        =
        \max_{\pistrat\in\Pi}\sum_{\ell\in C(\pistrat)}w_\ell
        \le \alpha_{\boldsymbol{w}}(G_\hatpistrat).
\]
Conversely, \Cref{lem:independent-compatible} says that every independent set $S$ is compatible with some pure original strategy $\pistrat$.  Then $S\subseteq C(\pistrat)$, and the nonnegativity of $\boldsymbol{w}$ gives
\[
        \sum_{\ell\in S}w_\ell
        \le
        \sum_{\ell\in C(\pistrat)}w_\ell
        \le
        h_{\mathcal{X}_{\Gamma}}(\boldsymbol{v}).
\]
Therefore,
\[
        h_{\mathcal{X}_{\Gamma}}(\boldsymbol{v})=\alpha_{\boldsymbol{w}}(G_\hatpistrat).
\]
It follows that
\[
        \rho(\widehat{\boldsymbol{x}}^\hatpistrat)=
        \sup_{\boldsymbol{w}\ge0,
              \boldsymbol{w}\ne0}
        \frac{\sum_{\ell\in L(\hatpistrat)}w_\ell}{\alpha_{\boldsymbol{w}}(G_\hatpistrat)}.
\]
By homogeneity, this is the same as maximizing $\sum_{\ell\in L(\hatpistrat)} w_\ell$ subject to $\boldsymbol{w}\ge0$ and $\alpha_{\boldsymbol{w}}(G_\hatpistrat)\le1$.  The latter constraint is precisely
\[
        \sum_{\ell\in S}w_\ell\le1
        \quad\text{for every independent set }S\in\Ind(G_\hatpistrat),
\]
which is the dual LP of $\chi_f(G_\hatpistrat)$.
\end{proof}

\rstcorollaryCorCommonFractionalConsequences*
\begin{proof}
By \Cref{thm:fractional-formula}, $\VR(\widehat\Gamma/\Gamma)$ is the maximum of finitely many fractional chromatic numbers.  Each $\chi_f(G_\hatpistrat)$ is the optimum of a rational linear program. Hence, it is rational. This also shows that the payoff supremum is attained by the dual optimal weights in the proof of \Cref{thm:fractional-formula}, translated back by $v_\ell=w_\ell/p_c(\ell)$ on $L(\hatpistrat)$ and by $v_\ell=0$ outside it. Finally, every nonempty graph $G$ satisfies $1\le\chi_f(G)\le |V(G)|$, and $|V(G_\hatpistrat)|\le |L|$.
\end{proof}

\subsection{Value of perfect information}
\label{appendix:voi}

We continue with the proofs concerning the value of perfect information.

\rstlemmaLemCograph*

\begin{proof}
We prove a stronger recursive claim.  For an original information set $I$ and a set $S\subseteq I$ of nodes with the same previous player sequence, $G(I,S)$ denotes the subgraph of $G_\hatpistrat$ induced by the leaves in $L(\hatpistrat)$ that are reachable from some node in $S$. We claim that $G(I,S)$ is a cograph. The proof is by induction on the maximum remaining tree height below nodes in $S$.

At $I$, the perfect-information strategy $\hatpistrat$ may choose different actions at different nodes of $S$.  For each action $a\in A_I$, we write $S_a\defeq\{h\in S:\hatpistrat(h)=a\}$. Any leaf below $S_a$ and any leaf below $S_b$ with $a\ne b$ are adjacent, since their paths require different actions at the same original information set $I$. Thus, $G(I,S)$ is the join of the subgraphs generated by the nonempty action classes $S_a$.

For an action $a$ with $S_a \ne \emptyset$, after action $a$ is taken, chance moves are followed until either a terminal leaf is reached or the next player information set is reached.  The set $T_a$ consists of the terminal leaves reached before another player decision.  These leaves form an edgeless component: they share the same previous player sequence, take the same action $a$ at $I$, and have no later player choices. For every subsequent information set $J$, $S_{a,J}\subseteq J$ is the set of first player nodes in $J$ reachable from $S_a$ after action $a$, before any other player decision.  The subgraph below $S_a$ is the disjoint union of the edgeless graph on $T_a$ and the graphs $G(J,S_{a,J})$ for the nonempty sets $S_{a,J}$.  Indeed, two leaves in different next-information-set components cannot conflict later.  If their paths later met at a common information set, then that information set would contain histories with different remembered previous information sets, namely $J$ on one path and a different $J'$ on the other, contradicting perfect recall.  Terminal leaves in $T_a$ also create no cross edges, since they have no later player information sets.

By induction, each nonempty $G(J,S_{a,J})$ is a cograph, and the edgeless graph on $T_a$ is a cograph.  Cographs are closed under disjoint union and join, so $G(I,S)$ is a cograph.  Finally, the full graph $G_\hatpistrat$ is obtained as a disjoint union over the first player information-set components reachable from the root, together with any chance-only terminal leaves. We conclude that $G_\hatpistrat$ is a cograph.
\end{proof}

\rstlemmaLemCographFractional*

This is standard: cographs are perfect, and perfect graphs have a fractional chromatic number equal to the clique number~\citep{Golumbic04:Algorithmic,ScheinermanUllman13:Fractional}. We include the short direct proof for completeness.

\begin{proof}[Proof of~\Cref{lem:cograph-fractional}]
Every graph satisfies $\chi_f(G)\ge\omega(G)$, because a clique of size $r$ needs total fractional color weight at least $r$.  We prove the reverse inequality for cographs by induction over the cograph construction.  It is immediate for a single vertex.  For disjoint unions,
\[
        \chi_f(G\cup G')=\max\{\chi_f(G),\chi_f(G')\},
        \quad
        \omega(G\cup G')=\max\{\omega(G),\omega(G')\}.
\]
For joins,
\[
        \omega(G\vee G')=\omega(G)+\omega(G').
\]
Every independent set in $G\vee G'$ lies entirely in $G$ or entirely in $G'$, so the fractional coloring LP decomposes, giving
\[
        \chi_f(G\vee G')=\chi_f(G)+\chi_f(G').
\]
Thus, equality is preserved by both operations.
\end{proof}

\rsttheoremThmVoiInteger*
\begin{proof}
The first statement and thus the integrality of $\operatorname{VoPI}$ follows by combining \Cref{thm:fractional-formula,lem:cograph,lem:cograph-fractional}. The upper bound follows from \Cref{prop:unaccounted-leaves} below.
\end{proof}

\begin{proposition}\label{prop:unaccounted-leaves}
For any strategy $\hatpistrat$ of the perfect-information refinement
$\GammaPI$,
\[
\omega(G_{\hatpistrat}) \leq \left\lceil \frac{|L|}{3} \right\rceil.
\]
\end{proposition}

\begin{proof}
We prove that $|L| \geq 3\omega(G_{\hatpistrat})-2$. The desired bound then follows immediately since $\omega(G_{\hatpistrat})$ is an integer. The proof follows an analogous structure to and uses the same notation as that of \Cref{lem:cograph}.

We prove a stronger recursive claim.  For an original information set $I$ and a set $S\subseteq I$ of nodes with the same previous player sequence, $G(I,S)$ denotes the subgraph of $G_\hatpistrat$ induced by the leaves in $L(\hatpistrat)$ that are reachable from some node in $S$, and $L(I,S)$ denotes the terminal leaves descending from nodes in $S$ (not necessarily reached by $\hatpistrat$). We claim that, whenever $G(I,S)$ is nonempty, $|L(I, S)| \geq 3\omega(G(I,S))-2.$ The proof is by induction on the maximum remaining tree height below nodes in $S$.

For each action $a\in A_I$, we write $S_a \defeq\{h\in S:\hatpistrat(h)=a\}$. Let $A^+ := \{a \in A_I : S_a \neq \emptyset\}$ and $r := |A^+|$. For each $a \in A^+$, let $L_a$ be the set of all terminal leaves descending
from the nodes in $S_a$ after taking action $a$, and let $G_a$ be the subgraph of $G(I,S)$ induced by the reached leaves in $L_a$. As shown in the proof of \Cref{lem:cograph}, $G_a$ is the disjoint union of subgraphs that obey the inductive hypothesis. Since $|L_a|$ is the sum of the leaves under each of these subgraphs and $\omega(G_a)$ is the maximum over their clique numbers, this implies $|L_a| \geq 3\omega(G_a) - 2$. Further, since $G(I,S)$ is the join of $\{G_a\}_{a \in A^+}$, we have $\omega(G(I,S)) = \sum_{a \in A^+} \omega(G_a)$.

It remains to count the leaves that are not reached by $\hatpistrat$. For every distinct $a,b \in A^+$, choose some node $h_a \in S_a$. Since all nodes in $I$ have the same action set $A_I$, action $b$ is available at $h_a$. The subtree reached by taking action $b$ at $h_a$ contains at least one terminal leaf, which cannot be reached by $\hatpistrat$ as $\hatpistrat$ chooses action $a$ at $h_a$. Moreover, the leaves obtained in this way for distinct ordered pairs $(a,b)$ are distinct: they lie below different action branches of the tree. Thus, there are at least $r(r-1)$ additional off-strategy terminal leaves descending from $S$, so
\[
\begin{aligned}
|L(I,S)|
&\ge
\sum_{a \in A^+} |L_a| + r(r-1)\\
&\ge
\sum_{a \in A^+} (3\omega(G_a)-2) + r(r-1)\\
&=
3\omega(G(I,S)) + r^2 - 3r.
\end{aligned}
\]
For every integer $r \geq 1$, we have $r^2 - 3r \geq -2$; hence,
$|L(I,S)| \geq 3\omega(G(I,S))-2$, completing the induction. The full graph $G_\hatpistrat$ is obtained as a disjoint union over the first player information-set components reachable from the root, together with any chance-only terminal leaves. Thus, $|L|$ is the sum of the leaves under each of these subgraphs, while $\omega(G_{\hatpistrat})$ is their maximum. This gives the desired upper bound.
\end{proof}

\rstpropositionPropVoiInNp*
\begin{proof}
By \Cref{thm:voi-integer}, $\operatorname{VoPI}(\Gamma)\ge \tau$ if and only if there exist a pure perfect-information strategy $\hatpistrat$ and a clique $C$ in $G_\hatpistrat$ with $|C|\ge\lceil \tau\rceil$.  A certificate lists $\hatpistrat$. Since $G_\hatpistrat$ is a cograph (\Cref{lem:cograph}), one can compute in polynomial time the clique number of $G_\hatpistrat$, and certify whether $\omega(G_\hatpistrat)$ exceeds the threshold $\tau$.
\end{proof}

\begin{restatable}{lemma}{rstlemmaLemCoverageGamePr}
\label{lem:coverage-decision-problem-pr}
The decision problem $\Gamma_{\mathcal{C}}$ has perfect recall.
\end{restatable}

\begin{proof}
At the first information set $I$, the player has no previous information sets or actions. At every information set $J_{a,j}$, the player remembers having been at $I$, having chosen $a$ as the first-stage action, and having observed $j$. This sequence is the same for all nodes $h_{e,a,j}$ in $J_{a,j}$. Hence, every information set consists of nodes with the same sequence.
\end{proof}

\rsttheoremThmVoiCoverageExact*
\begin{proof}
We prove each inequality separately.

\medskip
\noindent\emph{Lower bound.}
We choose indices $j_1,\ldots,j_k$ attaining $\OPT$ and write
\[
        C\defeq\bigcup_{a=1}^k S_{j_a}.
\]
For every $e\in C$, we choose one index $\alpha(e)\in[k]$ with $e\in S_{j_{\alpha(e)}}$, and set
\[
        C_a\defeq\{e\in C:\alpha(e)=a\}.
\]
We define a nonnegative payoff vector $\boldsymbol{v}$ as
\[
        v_{\ell(e,a,j,b)}\defeq
        \begin{cases}
        \displaystyle \frac{1}{p_{e,j_a}}
        & \text{if } e\in C_a,
          \ a=\alpha(e),
          \ j=j_a,
          \ b=e,\\[2mm]
        0 & \text{otherwise.}
        \end{cases}
\]
Since $p_{e,j_a}=1/(n d_e)$, the positive payoff is $nd_e$.  There is exactly one positive-payoff leaf for each covered element $e\in C$, namely $\ell_e\defeq \ell(e,\alpha(e),j_{\alpha(e)},e)$. In the perfect-information refinement, the player observes $e$ at $h_e$.  If $e\in C$, the player in the refinement chooses $\alpha(e)$ as the first-stage action. At the node $h_{e,\alpha(e),j_{\alpha(e)}}$, the player in the refinement chooses $b=e$ as the second-stage action.  Then for every covered element $e$, whenever chance first chooses $e$ and then chooses $j_{\alpha(e)}$, the play reaches $\ell_e$, whose contribution weighted by the chance probability is
\[
        p_{e,j_{\alpha(e)}}\cdot
        \frac{1}{p_{e,j_{\alpha(e)}}}=1.
\]
Thus, $V_{\Gamma_{\mathcal{C}}^{\operatorname{PI}}}(\boldsymbol{v})\ge |C|=\OPT$. (In fact, there are only $\OPT$ leaves with positive payoff, and each can contribute at most one after accounting for the chance probability, so $V_{\Gamma_{\mathcal{C}}^{\operatorname{PI}}}(\boldsymbol{v})=\OPT$.)

For any behavioral strategy in the original decision problem, $b_I(a)$ denotes the probability of first-stage action $a$ at $I$, and $b_{a,j}(b)$ denotes the probability of second-stage action $b$ at $J_{a,j}$.  The expected payoff is
\begin{align*}
        \mathbb E_\Gamma[\boldsymbol{v}]
        &=
        \sum_{a=1}^k\sum_{e\in C_a}
        b_I(a)p_{e,j_a}b_{a,j_a}(e)
        \cdot \frac{1}{p_{e,j_a}}  \\
        &=
        \sum_{a=1}^k b_I(a)\sum_{e\in C_a}b_{a,j_a}(e)
        \le
        \sum_{a=1}^k b_I(a)=1,
\end{align*}
where the inequality uses $\sum_{e\in C_a} b_{a,j_a}(e) \le \sum_{e \in U}b_{a,j_a}(e)=1$ for each fixed $a$. (This bound is tight: for any nonempty $C_a$ and any $e\in C_a$, the original player can play $a$ as the first-stage action and then $b=e$ as the second-stage action at $J_{a,j_a}$, obtaining a contribution of exactly one, so $V_{\Gamma_{\mathcal{C}}}(\boldsymbol{v})=1$.) Therefore,
\[
        \operatorname{VoPI}(\Gamma_{\mathcal{C}})\ge
        \frac{V_{\Gamma_{\mathcal{C}}^{\operatorname{PI}}}(\boldsymbol{v})}{V_{\Gamma_{\mathcal{C}}}(\boldsymbol{v})}=\OPT.
\]

\medskip
\noindent\emph{Upper bound.}
For an arbitrary $\boldsymbol{v}\ge0$, it suffices to show
\[
        V_{\Gamma_{\mathcal{C}}^{\operatorname{PI}}}(\boldsymbol{v})
        \le \OPT V_{\Gamma_{\mathcal{C}}}(\boldsymbol{v}).
\]
For a fixed pure perfect-information strategy $\hatpistrat$, $a(e)\in[k]$ is the first-stage action chosen by $\hatpistrat$ at $h_e$, and this defines the partition $P_a\defeq\{e\in U:a(e)=a\}$. For each pair $(a,j)$, we define
\[
        M_{a,j}\defeq\max_{b\in U}\sum_{e\in S_j}
        p_{e,j}v_{\ell(e,a,j,b)}.
\]
If the original player commits to first-stage action $a$, then at each later information set $J_{a,j}$ it can choose the best $b$. Thus,
\[
        V_{\Gamma_{\mathcal{C}}}(\boldsymbol{v})
        =
        \max_{a\in[k]}\sum_{j=1}^m M_{a,j}.
\]
At a singleton node $h_{e,a(e),j}$ in the perfect-information refinement, $\hatpistrat$ may choose a second-stage action depending on $e$. Therefore, the contribution of $\hatpistrat$ from a fixed pair $(a,j)$ is at most
\[
        \sum_{e\in P_a\cap S_j}
        \max_{b\in U}p_{e,j}v_{\ell(e,a,j,b)}.
\]
For each individual $e\in S_j$,
\[
        \max_{b\in U}p_{e,j}v_{\ell(e,a,j,b)}
        \le M_{a,j}.
\]
Hence,
\begin{align*}
        V^{\hatpistrat}_{\Gamma_{\mathcal{C}}^{\operatorname{PI}}}(\boldsymbol{v})
        &\le
        \sum_{a=1}^k\sum_{j=1}^m
        |P_a\cap S_j|M_{a,j}.
\end{align*}
If we write $c_a\defeq\max_{j\in[m]}|P_a\cap S_j|$, then
\[
        V^{\hatpistrat}_{\Gamma_{\mathcal{C}}^{\operatorname{PI}}}(\boldsymbol{v})
        \le
        \sum_{a=1}^k c_a\sum_{j=1}^mM_{a,j}
        \le\left(\sum_{a=1}^k c_a\right)V_{\Gamma_{\mathcal{C}}}(\boldsymbol{v}).
\]
For each $a$, we choose $j_a$ attaining $c_a=|P_a\cap S_{j_a}|$.  Since $P_1,\ldots,P_k$ form a partition, the sets $P_a\cap S_{j_a}$ are pairwise disjoint. Therefore,
\[
        \sum_{a=1}^k c_a
        =\sum_{a=1}^k |P_a\cap S_{j_a}|
        \le
        \left|\bigcup_{a=1}^k S_{j_a}\right|
        \le \OPT.
\]
Thus, $V^{\hatpistrat}_{\Gamma_{\mathcal{C}}^{\operatorname{PI}}}(\boldsymbol{v})\le \OPT V_{\Gamma_{\mathcal{C}}}(\boldsymbol{v})$.  Since this holds for every pure perfect-information strategy $\hatpistrat$ and $\boldsymbol{v}\ge0$ was arbitrary, $\operatorname{VoPI}(\Gamma_{\mathcal{C}})\le \OPT$. This concludes the proof.
\end{proof}

\rsttheoremThmVoiNpComplete*
\begin{proof}
Membership was established in~\Cref{prop:voi-in-np}. For hardness, we map a \textsc{MaximumCoverage} instance with threshold $\tau$ to $\Gamma_{\mathcal{C}}$ with the same threshold $\tau$.  By \Cref{thm:voi-coverage-exact}, $\operatorname{VoPI}(\Gamma_{\mathcal{C}})=\OPT$, so $\operatorname{VoPI}(\Gamma_{\mathcal{C}})\ge \tau$ if and only if the \textsc{MaximumCoverage} instance covers at least $\tau$ elements using $k$ sets. Moreover, the reduction is polynomial.
\end{proof}

\subsection{Value of recall}

We next include the omitted proofs concerning the value of recall.

\begin{lemma}
\label{lem:poly-bit-witness}
For $\tau\in\mathbb{Q}_{>0}$, if $\operatorname{VoR}(\Gamma)\ge \tau$, then there exist a refined pure strategy $\hatpistrat\in\Pi(\GammaPR)$ and a nonnegative rational payoff vector $\boldsymbol{v}\in\mathbb{Q}^L_{\ge0}$ such that $\sum_{\ell\in L}v_\ell=1$ and $\langle \boldsymbol{v},\widehat{\boldsymbol{x}}^\hatpistrat\rangle \ge \tau\langle \boldsymbol{v},\boldsymbol{x}^\pistrat\rangle$ for every $\pistrat\in\Pi$, with encoding length polynomial in the input size and the encoding length of $\tau$.
\end{lemma}

\begin{proof}
By homogeneity, the payoff supremum can be restricted to the simplex $\Delta_L$. For every $\boldsymbol{v}\in\Delta_L$, $V_\Gamma(\boldsymbol{v})>0$: some leaf has $v_\ell>0$, every live leaf has $p_c(\ell)>0$, and a pure strategy can choose the player actions on the path to that leaf. Thus, the ratio $V_{\GammaPR}(\boldsymbol{v})/V_\Gamma(\boldsymbol{v})$ is continuous on the compact simplex $\Delta_L$, and its supremum is a maximum. If $\operatorname{VoR}(\Gamma)\ge \tau$, we choose a maximizing payoff vector $\boldsymbol{v}\in\Delta_L$ and a refined pure strategy $\hatpistrat$ attaining $V_{\GammaPR}(\boldsymbol{v})$. Then
\[
        \langle \boldsymbol{v},\widehat{\boldsymbol{x}}^\hatpistrat\rangle
        \ge \tau\langle \boldsymbol{v},\boldsymbol{x}^\pistrat\rangle
        \quad\forall\pistrat\in\Pi.
\]
For this fixed $\hatpistrat$, the feasible set is a nonempty rational polytope in $|L|$ variables, defined by nonnegativity, normalization, and the linear inequalities above. A nonempty rational polytope has a basic feasible solution. Such a solution is determined by at most $|L|$ linearly independent tight constraints, and Cramer's rule bounds the bit length of the solution polynomially in the bit lengths of the coefficients and in $|L|$.  The coefficients are products of rational chance probabilities along root-to-leaf paths and the rational threshold $\tau$, so their bit lengths are polynomial in the input size. We conclude that a rational feasible payoff vector with polynomial encoding length exists.
\end{proof}

\rstpropositionPropVorInSigmatwo*
\begin{proof}
A $\class{\Sigma^p_2}$ verifier existentially guesses a refined pure strategy $\hatpistrat$ and the polynomial-bit rational vector $\boldsymbol{v}$ from \Cref{lem:poly-bit-witness}.  It then universally guesses an original pure strategy $\pistrat$ and checks in polynomial time that $\langle \boldsymbol{v},\widehat{\boldsymbol{x}}^\hatpistrat\rangle \ge \tau\langle \boldsymbol{v},\boldsymbol{x}^\pistrat\rangle$. The leaf-distribution vectors induced by the guessed pure strategies are computable by traversing the explicit tree and multiplying rational chance probabilities. Hence, the predicate after the existential and universal guesses is polynomial-time decidable.
\end{proof}

\subsubsection{\NP-hardness through fractional coloring embedding}
\label{app:vor-graph-embedding}

\begin{lemma}
\label{lem:edge-spine-size}
For $n=|V|$ and $m=|E|$, the decision problem $\Gamma_G$ has one chance node, $n+2m$ player nodes, $n+2m$ terminal leaves, and $n+m$ information sets.  In particular, it can be constructed in time $O(n+m)$. Moreover, $\Gamma_G$ has no absentmindedness.
\end{lemma}

\begin{proof}
There is one \emph{marker} node for each vertex.  By the handshaking identity, $\sum_{v\in V}d(v)=2m$, so there are exactly $2m$ edge-test nodes.  Each vertex contributes one success leaf, and each edge-test node contributes one exit leaf, for a total of $n+2m$ terminal leaves. The marker information sets are $M_v$ for $v\in V$, and the edge-test information sets are $I_e$ for $e\in E$, giving $n+m$ information sets. Finally, no root-to-leaf path visits an information set more than once, so the decision problem has no absentmindedness.
\end{proof}

\begin{lemma}
\label{lem:edge-spine-pr}
For every edge $e=\{u, v\}$, the information set $I_e$ contains exactly two player nodes, one on the $u$-branch and one on the $v$-branch, and these two nodes lie in distinct information sets of $\GammaPR$. As a result, a pure strategy of $\GammaPR$ can choose the actions at the two endpoint occurrences of every $I_e$ independently.
\end{lemma}

\begin{proof}
The information set $I_e$ occurs on precisely the branches corresponding to the endpoints of $e$, and once on each such branch, by \Cref{def:graph-embedding-construction}.  Let these nodes be $h_{u,e}$ and $h_{v,e}$.  Before reaching $h_{u,e}$, the player has visited the singleton marker $M_u$ and chosen $\bot$; before reaching $h_{v,e}$, the player has instead visited $M_v$ and chosen $\bot$.  Since $u\ne v$, the original information sets $M_u$ and $M_v$ are distinct.  Therefore, the remembered player sequences $\operatorname{seq}(h_{u,e})$ and $\operatorname{seq}(h_{v,e})$ are different.  By the definition of the coarsest perfect-recall refinement in \Cref{def:pr-refinement}, the two nodes are separated in $\GammaPR$.  Each refined part is a singleton because $I_e$ has only these two nodes.  Choices at distinct refined information sets are independent components of a pure strategy.
\end{proof}

For a refined pure strategy $\hatpistrat\in\Pi(\GammaPR)$ and a vertex $v\in V$, let $\ell_v(\hatpistrat)$ be the unique terminal leaf reached conditional on the root chance outcome $v$.

\begin{lemma}
\label{lem:one-leaf-per-vertex}
For every $\hatpistrat\in\Pi(\GammaPR)$,
\[
        L(\hatpistrat)=\{\ell_v(\hatpistrat):v\in V\},
\]
and the map $v\mapsto\ell_v(\hatpistrat)$ is a bijection from $V$ to $L(\hatpistrat)$.
\end{lemma}

\begin{proof}
Once chance selects $v$, all subsequent moves are deterministic under the pure strategy $\hatpistrat$.  At each edge test, either the chosen action exits immediately or the unique continuing edge is followed; if no exit occurs, we will reach $\ell_v$. This means that exactly one terminal leaf below the $v$-branch is reached. Distinct root branches are disjoint subtrees, so the reached leaves for distinct vertices are distinct. Finally, every root chance outcome has probability $1/n>0$, so all these leaves belong to the support $L(\hatpistrat)$.
\end{proof}

We use the bijection in \Cref{lem:one-leaf-per-vertex} to map the conflict graph back to the original vertex set.  Namely, define ${G'}_{\hatpistrat}$ on vertex set $V$ by
\[
        \{u, v\}\in E({G'}_{\hatpistrat})
        \quad\Longleftrightarrow\quad
        \{\ell_u(\hatpistrat), \ell_v(\hatpistrat)\}\in E(G_{\hatpistrat}).
\]
Then ${G'}_{\hatpistrat}\cong G_{\hatpistrat}$.

\begin{lemma}
\label{lem:strategy-spanning-subgraph}
For every $\hatpistrat\in\Pi(\GammaPR)$, ${G'}_{\hatpistrat}$ is a spanning subgraph of $G$.
\end{lemma}

\begin{proof}
Both graphs have the same vertex set $V$, so what remains is to prove $E({G'}_{\hatpistrat})\subseteq E(G)$.  Suppose $\{u, v\}\in E({G'}_{\hatpistrat})$.  By the definition of the conflict graph in \Cref{def:conflict-graph}, there is an original information set $J$ visited by both root-to-leaf paths to $\ell_u(\hatpistrat)$ and $\ell_v(\hatpistrat)$ at which those paths use different actions. A marker information set cannot be such a $J$, so $J=I_e$ for some edge $e\in E$. By construction, the only root branches containing nodes of $I_e$ are the two endpoint branches of $e$.  Since the two paths under consideration lie on the $u$- and $v$-branches, respectively, it follows that $e=\{u, v\}$.  Thus, $\{u, v\}\in E(G)$, as required.
\end{proof}

\begin{lemma}
\label{lem:all-continuing-realizes}
There is a refined pure strategy $\hatpistrat^* \in\Pi(\GammaPR)$ for which $G_{\hatpistrat^*} \cong G$.
\end{lemma}

\begin{proof}
At the refined occurrence of $I_e$ on the branch of an endpoint, let it choose action $0$ on the tail branch of the oriented edge and action $1$ on the head branch. This is well-defined because the two endpoint occurrences of $I_e$ are separated in $\GammaPR$ by \Cref{lem:edge-spine-pr}. Every vertex branch therefore reaches its success leaf, so $L(\hatpistrat^*)=\{\ell_v:v\in V\}$.

Let $e$ be an oriented edge from $u$ to $v$. The root-to-$\ell_u$ path visits $I_e$ and uses action $0$, while the root-to-$\ell_v$ path visits the same original information set and uses action $1$.  Hence $\{\ell_u, \ell_v\}$ is an edge of $G_{\hatpistrat^*}$. This proves that every edge of $G$ appears in the conflict graph. Conversely, \Cref{lem:strategy-spanning-subgraph} applied to $\hatpistrat^*$ shows that no nonedge of $G$ can appear. The bijection $v\mapsto \ell_v$ is therefore a graph isomorphism from $G$ to $G_{\hatpistrat^*}$.
\end{proof}

We next point out a simple fact concerning the fractional chromatic number; the proof is immediate.

\begin{lemma}[Monotonicity under edge addition]
\label{lem:fractional-monotonicity}
If $G'$ is a spanning subgraph of a graph $G$, then $\chi_f(G')\le \chi_f(G)$.
\end{lemma}

\NPfractional*

\begin{proof}
The size and no-absentmindedness statements were established in~\Cref{lem:edge-spine-size}. Let us consider an arbitrary refined pure strategy $\hatpistrat$. By \Cref{lem:one-leaf-per-vertex}, $G_{\hatpistrat}\cong {G'}_{\hatpistrat}$, and by \Cref{lem:strategy-spanning-subgraph}, ${G'}_{\hatpistrat}$ is a spanning subgraph of $G$. Therefore, \Cref{lem:fractional-monotonicity} gives $\chi_f(G_{\hatpistrat}) = \chi_f({G'}_{\hatpistrat}) \le \chi_f(G)$. Since this holds for every $\hatpistrat\in\Pi(\GammaPR)$,
\begin{equation}
\label{eq:embedding-upper}
        \max_{\hatpistrat\in\Pi(\GammaPR)}\chi_f(G_{\hatpistrat})
        \le \chi_f(G).
\end{equation}
On the other hand, \Cref{lem:all-continuing-realizes} gives a refined pure strategy $\hatpistrat^*$ with $G_{\hatpistrat^*}\cong G$, and hence
\begin{equation}
\label{eq:embedding-lower}
        \max_{\hatpistrat\in\Pi(\GammaPR)}\chi_f(G_{\hatpistrat})
        \ge \chi_f(G_{\hatpistrat^*})
        =\chi_f(G).
\end{equation}
Combining \eqref{eq:embedding-upper} and \eqref{eq:embedding-lower}, we have
\[
        \operatorname{VoR}(\Gamma_G)
        =\max_{\hatpistrat\in\Pi(\GammaPR)}\chi_f(G_{\hatpistrat})
        =\chi_f(G).
\]
\end{proof}

\begin{corollary}
\label{cor:vor-np-hard}
Exact computation of value of recall is \NP-hard, even when restricted
to the linear-size family of decision problems with no absentmindedness
$\{\Gamma_G\}$ from \Cref{def:graph-embedding-construction}.
\end{corollary}

\begin{proof}
For a rational number $r > 0$, it is \NP-complete to decide whether a given finite graph $G$ satisfies $\chi_f(G)\le r$~\citep{ScheinermanUllman13:Fractional}. Given $G$, we can construct the decision problem $\Gamma_G$ from~\Cref{def:graph-embedding-construction}.  By \Cref{thm:graph-embedding}, this is a polynomial-time
reduction and satisfies $\operatorname{VoR}(\Gamma_G)=\chi_f(G)$. As a result, a polynomial-time algorithm computing
$\operatorname{VoR}(\Gamma)$ exactly would decide an \NP-complete problem.
\end{proof}

\subsubsection{\texorpdfstring{$\class{\Sigma^p_2}$-hardness}{Sigma-2-p hardness}}

We next establish $\class{\Sigma^p_2}$-hardness. To begin with, we provide a standard conversion to 3-CNF, which will be used for reducing to 3-colorability.

\begin{lemma}
\label{lem:restriction-preserving-3cnf}
For any CNF formula $\psi$ with no empty clause and any finite set $P$ of variables, one can construct in polynomial time a 3-CNF formula $\psi_3$ that contains every variable in $P$ and has the following property: for every partial truth assignment $\rho$ to $P$, the restricted formula $\psi|_\rho$ is satisfiable if and only if $\psi_3|_\rho$ is satisfiable.
\end{lemma}

\begin{proof}
First, for each variable $p\in P$ that does not occur in $\psi$, we can add the tautological binary clause $(p\vee\neg p)$.  This does not affect satisfiability under any restriction of $P$, and it ensures that every variable in $P$ occurs. Each clause is transformed independently, using new auxiliary variables for that clause only.  A unit clause $(\ell)$ is replaced by
\[
        (\ell\vee y\vee z)
        \wedge(\ell\vee y\vee\neg z)
        \wedge(\ell\vee\neg y\vee z)
        \wedge(\ell\vee\neg y\vee\neg z).
\]
For any fixed value of the old variables, this conjunction is satisfiable over $y,z$ if and only if $\ell$ is $\mathsf{TRUE}$.  A binary clause $(\ell_1\vee\ell_2)$ is replaced by
\[
        (\ell_1\vee\ell_2\vee y)
        \wedge
        (\ell_1\vee\ell_2\vee\neg y),
\]
which is satisfiable over $y$ if and only if $\ell_1\vee\ell_2$ is $\mathsf{TRUE}$.  A clause of length three is left unchanged.

For a clause $C=(\ell_1\vee\cdots\vee\ell_k)$ of length $k>3$, we introduce fresh variables $y_1,\ldots,y_{k-3}$ and replace $C$ by
\[
        (\ell_1\vee\ell_2\vee y_1)
        \wedge
        \bigwedge_{j=2}^{k-3}(\neg y_{j-1}\vee\ell_{j+1}\vee y_j)
        \wedge
        (\neg y_{k-3}\vee\ell_{k-1}\vee\ell_k),
\]
where the middle conjunction is empty when $k=4$.  If all literals $\ell_i$ are $\mathsf{FALSE}$, the first clause forces $y_1=\mathsf{TRUE}$, the next forces $y_2=\mathsf{TRUE}$, and so on, making the last clause $\mathsf{FALSE}$.  If some $\ell_i$ is $\mathsf{TRUE}$, the auxiliary variables can be chosen to propagate the value $\mathsf{TRUE}$ only until the first $\mathsf{TRUE}$ literal is reached and then set the remaining $y$'s $\mathsf{FALSE}$; this satisfies the chain. Hence, the chain is satisfiable over the new variables exactly when $C$ is $\mathsf{TRUE}$. Since all auxiliary variables are new and no variable of $P$ is altered, satisfiability under every restriction of $P$ is preserved.
\end{proof}

\rstlemmaLemOptionalEdgeThreecol*
\begin{proof}
Let $\overline{\Phi}(X,Y)\defeq\neg\Phi(X,Y)$. This is a CNF formula. We apply \Cref{lem:restriction-preserving-3cnf} to $\overline{\Phi}$, taking $X$ as the initial variable set. This gives a $3$-CNF formula $\overline{\Phi}_3(X,Y,Z)$, where $Z$ is the set of new auxiliary variables, such that, for every $\boldsymbol{a}\in\{0,1\}^n$,
\[
\begin{aligned}
        \overline{\Phi}_3(\boldsymbol{a},Y,Z)
        \text{ is satisfiable over }Y\cup Z
        \quad\Longleftrightarrow\quad
        \overline{\Phi}(\boldsymbol{a},Y)
        \text{ is satisfiable over }Y.
\end{aligned}
\]
We now build a reduction from $\overline{\Phi}_3$ to $3$-colorability. We add the triangle with edges $\{\mathsf T, \mathsf F\}$, $\{\mathsf F, \mathsf B\}$, and $\{\mathsf B, \mathsf T\}$.  A vertex adjacent to $\mathsf B$ will be restricted to the two truth colors, namely the color of $\mathsf T$ and the color of $\mathsf F$.  For every variable $z$ of $\overline{\Phi}_3$, we add two vertices $z$ and $\bar z$, make both adjacent to $\mathsf B$, and add the edge $\{z, \bar z\}$. This means that exactly one of $z,\bar z$ has the color of $\mathsf T$ in any proper $3$-coloring.  We interpret this as the truth value of $z$.

We next use a \emph{not-all-equal} gadget.  Namely, given three terminal vertices $x,y,z$, each already restricted to truth colors, we add three new vertices $x',y',z'$, make $x',y',z'$ a triangle, and add the edges $\{x, x'\}$, $\{y, y'\}$, and $\{z, z'\}$.  This gadget is colorable, extending the terminal colors, if and only if the three terminal colors are not all equal.  Indeed, if the terminals are all equal, then the primed triangle has only two available colors.  If the terminal colors are not all equal, the primed vertex adjacent to a minority terminal receives the majority truth color, and the other two primed vertices receive the remaining truth color and the color of $\mathsf B$.

For a $3$-clause $(\ell_1\vee\ell_2\vee\ell_3)$, where each literal is represented by its literal vertex, we introduce a new vertex $w$ and its negation $\bar w$, both adjacent to $\mathsf B$, with the edge $\{w, \bar w\}$.  We add two not-all-equal gadgets: $\mathrm{NAE}(\ell_1,\ell_2,w)$ and $\mathrm{NAE}(\bar w,\ell_3,\mathsf F)$.

For fixed literal truth values, there exists a truth color for $w$ satisfying both gadgets if and only if at least one of $\ell_1,\ell_2,\ell_3$ is $\mathsf{TRUE}$.  If all three literals are $\mathsf{FALSE}$, the first gadget forces $w$ to be $\mathsf{TRUE}$, in which case $\bar w$ is $\mathsf{FALSE}$ and the second gadget sees three $\mathsf{FALSE}$ terminals.  Conversely, if the clause is $\mathsf{TRUE}$, we choose $w=\mathsf{FALSE}$ unless $\ell_1=\ell_2=\mathsf{FALSE}$; in that case $\ell_3=\mathsf{TRUE}$, and choosing $w=\mathsf{TRUE}$ satisfies both gadgets. Thus, before optional edges are added, the graph is $3$-colorable exactly when $\overline{\Phi}_3$ is satisfiable. $G_0$ is the graph obtained from all variable and clause gadgets, before adding any optional edges. Now, for every existential variable $x_i$, we define
\[
        \lambda_i^{(0)}\defeq x_i
        \text{ and }
        \lambda_i^{(1)}\defeq \bar x_i,
\]
where the right-hand sides denote the corresponding literal vertices. Further,
\[
        e_i^{(0)}\defeq\{x_i,\mathsf T\}
        \text{ and }
        e_i^{(1)}\defeq\{\bar x_i,\mathsf T\}.
\]

Let us fix $\boldsymbol{a}\in\{0,1\}^n$.  If $a_i=0$, the optional edge $e_i^{(0)}$ prevents the vertex $x_i$ from receiving the color of $\mathsf T$.  Since $x_i$ is also adjacent to $\mathsf B$, it must receive the color of $\mathsf F$, so the variable $x_i$ is $\mathsf{FALSE}$.  If $a_i=1$, the optional edge $e_i^{(1)}$ similarly forces $\bar x_i$ to have the color of $\mathsf F$, and therefore forces $x_i$ to have the color of $\mathsf T$, so $x_i$ is $\mathsf{TRUE}$.  As a result, the selected optional edges enforce $X=\boldsymbol{a}$. We have thus shown that
\[
\begin{aligned}
        G_{\boldsymbol{a}}\text{ is $3$-colorable}
        &\Longleftrightarrow
        \exists Y,Z\,
        \overline{\Phi}_3(\boldsymbol{a},Y,Z) \\
        &\Longleftrightarrow
        \exists Y\,
        \overline{\Phi}(\boldsymbol{a},Y) \\
        &\Longleftrightarrow
        \exists Y\,
        [\Phi(\boldsymbol{a},Y)=\mathsf{FALSE}].
\end{aligned}
\]
Negating this equivalence gives
\[
        G_{\boldsymbol{a}}\text{ is not $3$-colorable}
        \quad\Longleftrightarrow\quad
        \forall Y\,\Phi(\boldsymbol{a},Y).
\]
The entire construction is polynomial, concluding the proof.
\end{proof}

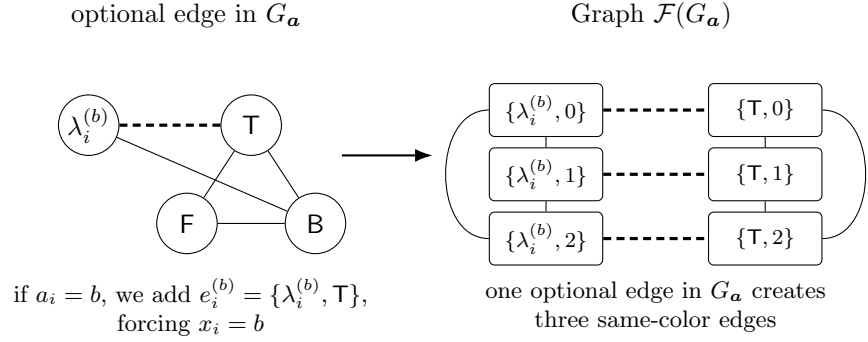
\begin{figure}[t]
\centering
\begin{tikzpicture}[
    x=1cm,y=1cm,
    >=Latex,
    palette/.style={circle,draw,minimum size=8mm,inner sep=0pt,font=\small,fill=white},
    literal/.style={circle,draw,minimum size=8mm,inner sep=0pt,font=\small,fill=white},
    fvertex/.style={rectangle,draw,rounded corners=2pt,minimum width=15mm,minimum height=7mm,align=center,font=\scriptsize,fill=white},
    box/.style={draw,rounded corners=2pt,align=center,inner xsep=7pt,inner ysep=5pt,font=\small},
    arrow/.style={->,thick},
    note/.style={align=center,font=\footnotesize}
]
\node[box,fill=white,minimum width=43mm,draw=none] (lefttitle) at (0,2.1) {
    optional edge in $G_{\boldsymbol{a}}$
};
\node[palette] (T) at (0.8,0.65) {$\mathsf T$};
\node[palette] (F) at (-0.05,-0.65) {$\mathsf F$};
\node[palette] (B) at (1.65,-0.65) {$\mathsf B$};
\draw (T) -- (F) -- (B) -- (T);
\node[literal] (lam) at (-1.35,0.65) {$\lambda_i^{(b)}$};
\draw (lam) -- (B);
\draw[densely dashed,very thick] (lam) -- (T);
\node[note] at (0,-1.75) {
    if $a_i=b$, we add $e_i^{(b)}=\{\lambda_i^{(b)}, \mathsf T\}$, \\
    forcing $x_i=b$
};

\node[box,fill=white,minimum width=46mm,draw=none] (righttitle) at (6.15,2.1) {
    Graph $\mathcal{F}(G_{\boldsymbol{a}})$
};
\foreach \c/\y in {0/0.85,1/0,2/-0.85} {
    \node[fvertex] (lc\c) at (4.7,\y) {$\{\lambda_i^{(b)},\c\}$};
    \node[fvertex] (tc\c) at (7.6,\y) {$\{\mathsf T,\c\}$};
    \draw[densely dashed,very thick] (lc\c) -- (tc\c);
}
\draw (lc0) -- (lc1) -- (lc2);
\draw (lc2) to[out=180,in=180,looseness=1.15] (lc0);
\draw (tc0) -- (tc1) -- (tc2);
\draw (tc2) to[out=0,in=0,looseness=1.15] (tc0);
\node[note] at (6.15,-1.75) {
    one optional edge in $G_{\boldsymbol{a}}$ creates\\
    three same-color edges
};

\draw[arrow] (2.,0.25) -- (3.2,0.25);

\end{tikzpicture}
\caption{The optional edge and the corresponding lift per~\Cref{lem:fglss-gap}. Choosing selector bit $a_i=b$ adds the optional edge incident to $\lambda_i^{(b)}$, which forces $x_i=b$ in the $3$-colorability encoding of $\neg\Phi$. In the FGLSS graph, that single choice induces three conflict edges, one for each color label. If $G_{\boldsymbol{a}}$ is $3$-colorable, then $\chi_f(\mathcal{F}_{\boldsymbol{a}})=3$; if $G_{\boldsymbol{a}}$ is not $3$-colorable, then $\chi_f(\mathcal{F}_{\boldsymbol{a}})\ge 3+\frac{3}{r-1}>\tau_r$ (\Cref{lem:fglss-gap}).}
\label{fig:optional-edge-fglss}
\end{figure}

\rstlemmaLemFglssGap*
\begin{proof}
If $\phi:W\to\{0,1,2\}$ is a proper 3-coloring of $G$, then
\[
        (w,c)\longmapsto c+\phi(w)\pmod 3
\]
is a proper 3-coloring of $\mathcal{F}(G)$.  Since the three vertices $(w,0),(w,1),(w,2)$ form a triangle for every $w$, the fractional chromatic number is exactly $3$.

If $G$ is not 3-colorable, every independent set of $\mathcal{F}(G)$ has size at most $r-1$.  An independent set of size $r$ would contain exactly one vertex $(w,c_w)$ for every $w\in W$, and the second edge rule would make $w\mapsto c_w$ a proper 3-coloring of $G$.  Thus, $\alpha(\mathcal{F}(G))\le r-1$.  For every graph $G'$, summing the vertex-covering constraints in the fractional-coloring LP gives
\[
        |V(G')|\le \alpha(G')\sum_{S\in\Ind(G')}z_S,
\]
and therefore $\chi_f(G')\ge |V(G')|/\alpha(G')$.  Since $|V(\mathcal{F}(G))|=3r$, the claimed lower bound follows.
\end{proof}

We conclude the proof by realizing the optional-edge graph family through the use of conflict graphs. We fix $G_0=(W,E_0)$ and optional edges $e_i^{(b)}=\{\lambda_i^{(b)}, \mathsf T\}$ from \Cref{lem:optional-edge-3col}.  For $\boldsymbol{a}\in\{0,1\}^n$, we write
\[
        \mathcal{F}_{\boldsymbol{a}}\defeq \mathcal{F}(G_{\boldsymbol{a}})
        \text{ and }
        G_{\boldsymbol{a}}\defeq G_0+\{e_i^{(a_i)}:i=1,\ldots,n\}.
\]
We construct a decision problem $\Gamma$ with no absentmindedness satisfying
\[
        \max_{\hatpistrat\in\Pi(\GammaPR)}\chi_f(G_\hatpistrat)
        =
        \max_{\boldsymbol{a}\in\{0,1\}^n}\chi_f(\mathcal{F}_{\boldsymbol{a}}).
\]
We write $Q \defeq W \times\{0,1,2\}$; these are the vertices of every $\mathcal{F}_{\boldsymbol{a}}$. The root of $\Gamma$ is a chance node choosing each $q\in Q$ with probability $1/(3r)$. For every fixed edge $e=\{q,q'\}$ of $\mathcal{F}(G_0)$, we create an original information set $I_e$ with action set $\{0,1\}$.  We orient $e$ arbitrarily.  At an occurrence of $I_e$ on the branch labeled by the tail endpoint, action $0$ is the distinguished continuing action; at an occurrence on the branch labeled by the head endpoint, action $1$ is the distinguished continuing action.  The other action goes immediately to a terminal leaf.

For every optional edge $e_i^{(b)}=\{\lambda_i^{(b)}, \mathsf T\}$ and every color $c\in\{0,1,2\}$, we create the optional edge $e_{i,b,c} \defeq\{(\lambda_i^{(b)},c), (\mathsf T,c)\}$ and an original information set $I_{i,b,c}$ with action set $\{0,1\}$.  At the occurrence on branch $(\lambda_i^{(b)},c)$, action $0$ is distinguished; at the occurrence on branch $(\mathsf T,c)$, action $1$ is distinguished.  Again, non-distinguished actions terminate immediately.

For every $i$, we create a selector information set $S_i$ with actions $\{0,1\}$.  This information set is placed as the first player information set on each branch $(\lambda_i^{(0)},c)$ and $(\lambda_i^{(1)},c)$, for $c=0,1,2$.  If the selector action is $b$, then branch $(\lambda_i^{(b)},c)$ later visits $I_{i,b,c}$, while branch $(\lambda_i^{(1-b)},c)$ bypasses its corresponding optional information set.  Branches not labeled by distinguished literal vertices have no selector.

After the selector, if there is one, every continuation visits a singleton \emph{private tag} information set with one dummy action. The tag is unique to the root chance label $q$ and, for distinguished literal branches, to the selector action just taken.  Branches without a selector receive a tag unique to $q$.  These tags ensure that the perfect-recall refinement remembers the branch before later shared tests are encountered. After the private tag, the branch visits conflict-test information sets in a fixed global order.  On branch $q$, the list consists of
\begin{enumerate}[label=(\roman*)]
\item all fixed information sets $I_e$ with $q\in e$;
\item if $q=(\mathsf T,c)$, all optional information sets $I_{i,b,c}$ for all $i$ and $b$;
\item if $q=(\lambda_i^{(b)},c)$ and the selector action on this branch was $b$, the optional information set $I_{i,b,c}$.
\end{enumerate}
At every conflict-test information set, the distinguished action continues to the next test and the other action terminates immediately.  If all distinguished actions on the branch are chosen, play terminates at a success leaf.  Thus, under every pure strategy of $\GammaPR$ and every root chance outcome $q$, exactly one terminal leaf below $q$ is reached.

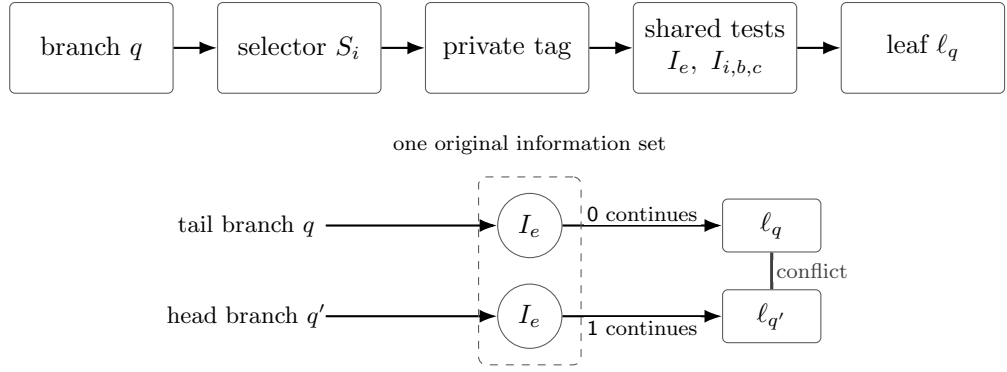
\begin{figure}[t]
\centering
\begin{tikzpicture}[
    x=1cm,y=1cm,
    >=Latex,
    stage/.style={draw=black!65,rounded corners=2pt,align=center,inner xsep=3pt,inner ysep=4pt,font=\small,minimum width=20mm,minimum height=12mm,text width=19.5mm,fill=white},
    chance/.style={stage},
    selector/.style={stage},
    tag/.style={stage},
    test/.style={stage},
    leaf/.style={stage},
    smalltest/.style={circle,draw=black!65,fill=white,minimum size=8.5mm,inner sep=0pt,font=\small},
    smallleaf/.style={rectangle,draw=black!65,fill=white,rounded corners=2pt,minimum width=13mm,minimum height=7mm,inner sep=2pt,font=\small},
    infoset/.style={draw=black!65,dashed,rounded corners=3pt,inner xsep=7pt,inner ysep=6pt},
    arrow/.style={->,thick},
    lane/.style={font=\footnotesize,align=right},
    labelbox/.style={font=\scriptsize,align=center,fill=white,inner sep=1.2pt},
    note/.style={font=\footnotesize,align=center,text width=42mm},
    result/.style={draw=black!55,rounded corners=2pt,align=center,font=\footnotesize,inner xsep=6pt,inner ysep=5pt,text width=40mm,fill=white}
]

\node[chance] (branch) at (0,0) {
    branch $q $
};
\node[selector] (selector) at (2.75,0) {
    selector $S_i$
};
\node[tag] (tag) at (5.50,0) {
    private tag
};
\node[test] (tests) at (8.25,0) {
    shared tests\\
    $I_e,\ I_{i,b,c}$
};
\node[leaf] (success) at (11.00,0) {
 leaf $\ell_q$
};

\draw[arrow] (branch) -- (selector);
\draw[arrow] (selector) -- (tag);
\draw[arrow] (tag) -- (tests);
\draw[arrow] (tests) -- (success);

\begin{scope}[xshift=2.05cm]
\node[lane] at (0.00,-2.35) {tail branch $q$};
\node[lane] at (0.00,-3.55) {head branch $q'$};
\node[smalltest] (hq) at (3.75,-2.35) {$I_e$};
\node[smalltest] (hqp) at (3.75,-3.55) {$I_e$};
\node[smallleaf] (lq) at (6.95,-2.35) {$\ell_q$};
\node[smallleaf] (lqp) at (6.95,-3.55) {$\ell_{q'}$};

\draw[arrow] (1.05,-2.35) -- (hq);
\draw[arrow] (1.05,-3.55) -- (hqp);
\draw[arrow] (hq) -- node[labelbox,above] {$\textsf{0}$ continues} (lq);
\draw[arrow] (hqp) -- node[labelbox,below] {$\textsf{1}$ continues} (lqp);

\node[infoset,fit=(hq)(hqp)] (Hbox) {};
\node[labelbox] at (3.75,-1.28) {one original information set};
\draw[black!75,very thick] (lq) -- node[labelbox,right] {conflict} (lqp);
\end{scope}

\end{tikzpicture}
\caption{Realizing the graph family as conflict graphs. Chance first selects a vertex $q \in Q$ of the graph. Selector information sets choose the optional-edge vector $\boldsymbol{a}$, unique tags separate histories in the perfect-recall refinement, and shared conflict-test information sets create precisely the edges of $\mathcal{F}_{\boldsymbol{a}}$ among the (success) leaves. In particular, different required actions at the same information set of the original game create exactly the edge $\{\ell_q, \ell_{q'}\}$. Optional tests $I_{i,b,c}$ are routed in only when the selector action satisfies $a_i=b$. }
\label{fig:vor-realization-decision-problem}
\end{figure}

\begin{lemma}
\label{lem:realization-decision-problem-structure}
The decision problem $\Gamma$ has no absentmindedness.  In its coarsest perfect-recall refinement, each selector $S_i$ remains one information set, and every occurrence of every conflict-test information set is split according to the branch and selector history leading to it. As a result, a pure strategy of $\GammaPR$ consists of selector bits $a_i\in\{0,1\}$, together with independent choices at all refined conflict-test occurrences.
\end{lemma}

\begin{proof}
Along any root-to-leaf path, the play visits at most one selector, then one private tag, and then a list of distinct conflict-test information sets.  If a non-distinguished action is taken, the path terminates immediately. Thus, no path visits the same original information set twice, so the decision problem has no absentmindedness.

All nodes in $S_i$ occur before the player has visited any information set, so they have the same remembered sequence and remain together in the coarsest perfect-recall refinement. A conflict-test occurrence is preceded by a private tag identifying the root branch and, where relevant, the selector action.  Different occurrences of the same original conflict-test information set therefore have different remembered sequences and are separated in $\GammaPR$.
\end{proof}

\vorreal*

\begin{proof}
First, for a fixed $\boldsymbol{a}\in\{0,1\}^n$, $\hatpistrat^{(\boldsymbol{a})}$ denotes the refined pure strategy that chooses selector action $a_i$ at $S_i$ and chooses the distinguished continuing action at every refined conflict-test occurrence. Then every chance branch reaches its success leaf, so the support leaves of $\hatpistrat^{(\boldsymbol{a})}$ can be indexed by $Q$. Two success leaves indexed by $q,q'\in Q$ conflict exactly when $q$ and $q'$ are adjacent in $\mathcal{F}_{\boldsymbol{a}}$. Moreover, fixed edges of the graph are realized by the corresponding original information sets $I_e$.  The optional edge $\{(\lambda_i^{(b)},c), (\mathsf T,c)\}$ is realized by $I_{i,b,c}$ exactly when $a_i=b$: the $\mathsf T$ branch always visits $I_{i,b,c}$, while the $(\lambda_i^{(b)},c)$ branch visits it precisely under selector action $b$.  Selector information sets create no conflict because all histories visiting $S_i$ require the same selector action $a_i$, and private tags have only one action.  Thus, $G_{\hatpistrat^{(\boldsymbol{a})}}\cong \mathcal{F}_{\boldsymbol{a}}$, giving
\[
        \max_{\hatpistrat}\chi_f(G_\hatpistrat)
        \ge
        \max_{\boldsymbol{a}}\chi_f(\mathcal{F}_{\boldsymbol{a}}).
\]

Conversely, consider an arbitrary refined pure strategy $\hatpistrat$.  The selector action chosen by $\hatpistrat$ at $S_i$ is denoted by $a_i$.  Since the root chance node has one branch for every $q\in Q$, and each branch reaches exactly one terminal leaf under $\hatpistrat$, the support $L(\hatpistrat)$ is in bijection with $Q$.  We map each reached terminal leaf to its root branch label.

Under this bijection, $G_\hatpistrat$ is a spanning subgraph of $\mathcal{F}_{\boldsymbol{a}}$.  A conflict cannot be caused by a selector, because all reached histories visiting $S_i$ require the same action $a_i$; tags have one action and are private. Hence, every conflict arises at a fixed information set $I_e$ or an optional information set $I_{i,b,c}$.  In the first case, $e$ is an edge of $\mathcal{F}(G_0)$ and hence of $\mathcal{F}_{\boldsymbol{a}}$.  In the second case, both endpoint leaves can visit $I_{i,b,c}$ only if the literal endpoint $(\lambda_i^{(b)},c)$ is routed there by selector action $a_i=b$; therefore the corresponding optional edge belongs to $\mathcal{F}_{\boldsymbol{a}}$. This shows that every edge of $G_\hatpistrat$ maps to an edge of $\mathcal{F}_{\boldsymbol{a}}$. Fractional chromatic number is monotone under edge addition. Therefore, $\chi_f(G_\hatpistrat)\le\chi_f(\mathcal{F}_{\boldsymbol{a}})$. Taking the maximum over $\hatpistrat$ gives the reverse inequality.
\end{proof}

\rsttheoremThmVorSigmatwoComplete*

\begin{proof}
Membership was established in~\Cref{prop:vor-in-sigma2}.  For hardness, we start with an instance $\exists X\,\forall Y\,\Phi(X,Y)$ of \textsc{ExistsForallDNFSAT}, which is $\class{\Sigma^p_2}$-complete~\cite{Stockmeyer76:Polynomial}. We may assume that every term of $\Phi$ is nonempty; if an empty term is present, the source instance is trivially YES and can be mapped to a one-branch chance-only decision problem with threshold $1$, a fixed YES instance of \textsc{VoR}-\textsc{Threshold}.

We apply~\Cref{lem:optional-edge-3col} to $\Phi$. This gives a graph $G_0=(W,E_0)$ and optional edges $e_i^{(b)}$ such that, for every $\boldsymbol{a}\in\{0,1\}^n$,
\[
        G_{\boldsymbol{a}}\text{ is not $3$-colorable}
        \quad\Longleftrightarrow\quad
        \forall Y\,\Phi(\boldsymbol{a},Y).
\]
For each $\boldsymbol{a}$, let $\mathcal{F}_{\boldsymbol{a}}\defeq \mathcal{F}(G_{\boldsymbol{a}})$. We construct the decision problem $\Gamma$ with no absentmindedness realizing this optional-edge graph family as in \Cref{lem:vor-graph-realization}.  Thus, by \Cref{thm:fractional-formula,lem:vor-graph-realization},
\[
        \operatorname{VoR}(\Gamma)
        =
        \max_{\boldsymbol{a}\in\{0,1\}^n}
        \chi_f(\mathcal{F}_{\boldsymbol{a}}).
\]
We let $r=|W|$ and set
\[
        \tau\defeq \tau_r
        =
        3+\frac{3}{2(r-1)}
        =
        \frac{6r-3}{2r-2}.
\]
By \Cref{lem:fglss-gap},
\[
        \chi_f(\mathcal{F}_{\boldsymbol{a}})\ge \tau
        \quad\Longleftrightarrow\quad
        G_{\boldsymbol{a}}\text{ is not $3$-colorable}.
\]
Indeed, if $G_{\boldsymbol{a}}$ is $3$-colorable, then $\chi_f(\mathcal{F}_{\boldsymbol{a}})=3<\tau$, while if it is not $3$-colorable, then $\chi_f(\mathcal{F}_{\boldsymbol{a}})\ge 3+3/(r-1)>\tau$. Combining these equivalences gives
\[
\begin{aligned}
        \operatorname{VoR}(\Gamma)\ge \tau
        &\Longleftrightarrow
        \exists\boldsymbol{a}\in\{0,1\}^n:
        \chi_f(\mathcal{F}_{\boldsymbol{a}})\ge \tau \\
        &\Longleftrightarrow
        \exists\boldsymbol{a}\in\{0,1\}^n:
        G_{\boldsymbol{a}}\text{ is not $3$-colorable} \\
        &\Longleftrightarrow
        \exists\boldsymbol{a}\in\{0,1\}^n\,
        \forall Y\,\Phi(\boldsymbol{a},Y).
\end{aligned}
\]
This is exactly the YES condition of the source instance. Moreover, the construction is polynomial, completing the proof.
\end{proof}

\end{document}